\documentclass[a4paper,10pt]{article}

\usepackage[T2A]{fontenc}

\usepackage[english]{babel}

\usepackage{graphicx}
\usepackage{amsthm, amsmath, amssymb, amsbsy, amsfonts, dsfont}
\usepackage{srcltx, color}

\newcommand{\PT}{\mathcal{PT}}

\def\tht{\theta}
\def\Om{\Omega}
\def\om{\omega}
\def\e{\varepsilon}
\def\g{\gamma}
\def\G{\Gamma}
\def\l{\lambda}
\def\p{\partial}
\def\D{\Delta}

\def\a{\alpha}
\def\b{\beta}

\def\d{\delta}
\def\L{\Lambda}

\def\vp{\varphi}

\def\vt{\vartheta}

\def\Ho{\mathring{H}_2}

\def\iu{\mathrm{i}}
\def\di{\,d}
\def\Ups{\Upsilon}

\def\Op{\mathcal{H}}

\def\hf{\mathfrak{h}}

\def\cL{\mathcal{L}}

\def\Ho{\mathring{W}_2}

\def\cA{\mathcal{A}}
\def\cR{\mathcal{R}}
\def\cB{\mathcal{B}}
\def\cG{\mathcal{G}}

\def\cI{\mathcal{I}}
\def\cP{\mathcal{P}}
\def\cU{\mathcal{U}}

\def\rE{\mathrm{E}}
\def\rM{\mathrm{M}}
\def\rS{\mathrm{S}}
\def\rJ{\mathrm{J}}

\def\ev{\mathrm{e}}

\def\Dom{\mathfrak{D}}

\DeclareMathOperator{\spec}{\sigma}

\DeclareMathOperator{\essspec}{\sigma_{ess}}

\DeclareMathOperator{\RE}{Re}
\DeclareMathOperator{\IM}{Im}
\DeclareMathOperator{\dist}{dist}
\DeclareMathOperator{\sgn}{sgn}

\numberwithin{equation}{section}

\newtheorem{theorem}{Theorem}[section]
\newtheorem{lemma}{Lemma}[section]

\begin{document}

\allowdisplaybreaks

\title{
Bifurcations of thresholds
in essential spectra of elliptic operators under localized non-Hermitian perturbations}

\date{\empty}

\author
{D. I. Borisov$^1$, D. A. Zezyulin$^2$\footnote{Corresponding author}, M. Znojil$^3$}

\vskip -0.5 true cm

\maketitle

\begin{center}
{\footnotesize $^1$
Institute of Mathematics, Ufa Federal Research Center, Russian Academy of Sciences, Ufa, Russia,
\\
\&
\\
Bashkir State 
University, 
Ufa, Russia, \\
\&
\\
University of Hradec Kr\'alov\'e,  Hradec Kr\'alov\'e, Czech Republic
 \\
{\tt borisovdi@yandex.ru}}
\end{center}

\begin{center}
{\footnotesize $^2$
ITMO University, Saint-Petersburg 197101, Russia
\\
{\tt d.zezyulin@gmail.com}}
\end{center}

\begin{center}
{\footnotesize $^3$
The Czech Academy of Sciences, Nuclear Physics Institute, Hlavn\'{\i} 130, 25068 \v{R}e\v{z}, Czech Republic
\\
\&
\\
University of Hradec Kr\'alov\'e,  Hradec Kr\'alov\'e, Czech Republic,
\\
\&
\\
Institute of System Science, Durban University of Technology, P.O. Box 1334, Durban 4000, South Africa
 \\
{\tt znojil@ujf.cas.cz}}
\end{center}

\begin{abstract}
We consider the operator
\begin{equation*}
  \Op=\Op' -\frac{\p^2\ }{\p x_d^2} \quad\text{on}\quad\om\times\mathds{R}
\end{equation*}
subject to the Dirichlet or Robin condition, where a domain $\om\subseteq\mathds{R}^{d-1}$ is  bounded or unbounded. The symbol $\Op'$ stands for a  second order self-adjoint  differential operator on $\om$ such that the spectrum of the operator $\Op'$ contains several discrete eigenvalues $\Lambda_{j}$, $j=1,\ldots, m$. These eigenvalues are thresholds in  the essential spectrum of the operator $\Op$. We study how these thresholds bifurcate once we add a small localized perturbation $\e\cL(\e)$ to the operator $\Op$, where $\e$ is a small positive parameter and $\cL(\e)$ is an abstract,  not necessarily symmetric operator.  We show that these thresholds bifurcate into eigenvalues and resonances   of the operator $\Op$ in the vicinity of $\L_j$ for sufficiently small $\e$. We prove effective simple conditions determining the existence of these resonances and eigenvalues and find the leading terms of their asymptotic expansions.  Our analysis  applies to  generic non-self-adjoint perturbations  and, in particular, to perturbations characterized by the  parity-time ($\PT$) symmetry. Potential applications of our result embrace  a broad class of  physical systems governed by dispersive or diffractive effects. As a case example, we employ our findings  to develop a   scheme for a controllable generation of non-Hermitian optical states with normalizable power and real part of the complex-valued propagation constant lying in the continuum. The corresponding eigenfunctions  can be interpreted as an optical generalization of   bound states {embedded} in the continuum. For a particular example, the persistence of asymptotic expansions is confirmed with direct numerical evaluation of the perturbed spectrum.
\end{abstract}

\section{Introduction}

\paragraph{Physical context and motivation.}  Physics of
 non-Hermitian  Hamiltonians is attracting steadily growing attention both
on the fundamental level in the development of complex formulations of quantum mechanics  \cite{Bender2}, \cite{Most}, \cite{Moiseyev} and
in several applied and experimental fields, such as optics and photonics, Bose-Einstein condensates of atoms or exciton-polaritons,  acoustics, and in other areas  where diffractive  or dispersive effects are governed by
Schr\"oginder-like elliptic operators,    see \cite{PukhovReview}, \cite{KYZreview}, \cite{FengReview}, \cite{LonghiReview}, \cite{El-GanainyReview} for recent reviews.  Prominent examples of essentially non-Hermitian phenomena that    were introduced in mathematical literature long ago but entered various areas of physics much more  recently include exceptional points \cite{Ka} and spectral singularities \cite{SS1}, \cite{SS2}, \cite{Vainberg}.  In particular, unusual effects associated with exceptional points
are being extensively discussed in   optics  and photonics    (e.g. \cite{Heiss}, \cite{Alu}, \cite{Nori}),  whereas   spectral singularities \cite{SS1}, \cite{SS2}, \cite{Vainberg}  are now understood to play an  important role in     wave  scattering   \cite{Mostafazadeh2009}, \cite{SS5} and are used to implement coherent perfect absorption of electromagnetic \cite{CPA1}, sound \cite{CPA2}, and matter \cite{CPA3} waves.   Another prototypical behavior, which is forbidden in Hermitian physics but  is of the utmost importance in non-Hermitian systems, is the transition from the entirely real spectrum of eigenvalues   to  a  partially complex one.

In a real-world   system,  the non-Hermiticity usually corresponds to the presence of  an  energy gain or absorption, which creates an effective complex potential for propagating waves \cite{Muga}.
An especially interesting situation, where
 a  judicious balance  between  amplification and losses results in rich physics, corresponds to
 the   so-called parity-time ($\PT$) symmetric systems
famous  for their property to robustly preserve   reality of all eigenvalues  in spite of the absence of Hermiticity \cite{Bender1,Bender2}. Tuning a control parameter of a $\PT$-symmetric system, one can realize various qualitative changes  in its   spectral structure, and these changes are  typically associated with rich and intriguing behaviors. The simplest of those behaviors, which was observed in a series of experiments  \cite{Alu,Nori}, corresponds to the collision of a pair of real discrete eigenvalues in an exceptional point with a subsequent splitting in a complex-conjugate pair. In systems with a continuous spectrum, the situation is further enriched. In particular, the  bifurcation of an isolated eigenvalue from the bottom of the essential spectrum can be   accompanied by
 a  so-called jamming anomaly, i.e., a  non-monotonous dependence of the energy flux  through the gain-to-loss interface  on the parameter characterizing the strength of  the  non-Hermiticity \cite{jamming}.  The bifurcations  of  a complex-conjugate pair of eigenvalues from an internal point  in
the essential spectrum are even more interesting and have recently been discussed as an unconventional mechanism of $\PT$-symmetry breaking \cite{Yang17}, \cite{Garmon}, \cite{KZ17}, \cite{KZ19}
distinctively   different from the better studied $\PT$-symmetry breaking   through an exceptional point.  Complex eigenfunctions associated with bifurcated eigenvalues are
 $L^2$-integrable, and  the real parts of these eigenvalues  belong to the continuous spectrum.  This
enables  a physical  interpretation  of such eigenfunctions in terms of    non-Hermitian generalizations \cite{BICoptics}, \cite{Kartashov} of bound states {embedded} in the continuum, well-known in quantum mechanics  \cite{BIC1}, \cite{BIC2},  \cite{BIC4},  \cite{Robnik},  optics, and other fields \cite{Soljacic}. It should be noticed at the same time that most of the activity  devoted to non-Hermitian optical bound states in the continuum is being carried out for one-dimensional systems. For  multi-dimensional geometries,  most of the available results are numerical in nature \cite{Kartashov}.   From the practical  point of view, it is also important that eigenfunctions associated with emerging from the essential spectrum eigenvalues are extremely weakly localized in the vicinity of the bifurcation, which   hinders their efficient numerical evaluation. Naturally, this problem is even more pronounced in multi-dimensional geometries, where 
 {many} more computational resources {are} necessary to approximate  the eigenfunctions. Therefore,  any analytical information on the properties of such states is highly desirable.

\paragraph{Mathematical context.} The phenomenon that a small localized perturbation of a self-adjoint differential operator can generate discrete eigenvalues from the edges in the essential spectrum is known {for} about a hundred of years. Its rigorous mathematical study was initiated by classical works by B.~Simon, M.~Klaus,  R.~Blankenbecler, M.~L.~Goldberger \cite{Si76}, \cite{Kl77}, \cite{BGS77}, \cite{KS80} and since that time, hundreds of papers on this subject were written. While classical works were devoted to the Schr\"odinger operator on an axis and plane perturbed by a small localized potentials, in further works the studies were made for
plenty of other models, like
waveguide-like structures, see, for instance, \cite{ESTV}, \cite{ED}, \cite{Na}, \cite{JPA07}, for periodic operators, see, for instance, \cite{Zhe}, for operators with distant perturbations \cite{MK} and many others. All these works treated symmetric perturbations of self-adjoint operators, and the perturbed operators were self-adjoint as well.

Non-symmetric perturbations of self-adjoint operators were studied in essentially less details. In \cite{Ga1}, \cite{Ga2} there was considered the Laplacian on the axis perturbed by a small abstract localized operator, which was not assumed to be symmetric. The main result was sufficient conditions ensuring the existence and absence of the emerging eigenvalues from the bottom of the essential spectrum and if they exist, the leading terms in the  asymptotic expansions for the emerging eigenvalues were found. These results were essentially extended in  \cite{Izv08}, \cite{Izv11}. Here an unperturbed operator was an arbitrary periodic self-adjoint operator on the line \cite{Izv08} or on the plane \cite{Izv11}. A perturbation was a small abstract operator not necessarily symmetric and localized in a much weaker sense than in  \cite{Ga1}, \cite{Ga2}. The structure of the spectra of such operators was studied in details. Qualitative properties like stability of the essential spectrum, the countability of the point spectrum, the absence of the residual spectrum, the existence of embedded eigenvalues were addressed. Sufficient conditions ensuring the existence and absence of the eigenvalues emerging from edges of internal
gaps in the essential spectrum were established  and if they exist, the leading terms in their asymptotic expansions were obtained. Eigenvalues emerging from the bottom of the essential spectrum were also studied in \cite{IEOP08}, \cite{MZ15}, \cite{MSb17} for waveguides with $\PT$-symmetric Robin-type boundary condition. In \cite{IEOP08}, a planar waveguide   was considered with a locally perturbed coefficient in the $\PT$-symmetric boundary condition. In \cite{MZ15}, \cite{MSb17}, similar two- and three-dimensional waveguides were considered and the perturbation was a small width of these waveguides. The main obtained results were sufficient conditions ensuring the existence of the emerging eigenvalues and the leading terms of their asymptotic expansions. We also mention work \cite{G3}, where the Dirichlet or Neumann Laplacian in a multi-dimensional  cylinder was considered and it was perturbed by a small localized non-symmetric perturbations. The eigenvalues bifurcating from the bottom and the internal thresholds in the essential spectrum were studied. There were obtained certain sufficient conditions ensuring the existence of such eigenvalues and the leading terms of their asymptotic expansions were calculated. However, there was {a 
gap} in calculations in \cite{G3}, which made the results of this work true {but 
incomplete. Namely, while working with operators providing meromorphic continuations for the resolvent in the vicinity of internal thresholds, the author of \cite{G3} considered only \textit{one} continuation, while, as we show in the present work, even in our more general setting \textit{two} continuations exist and complex eigenvalues are the poles just for one of these continuations. This is why the results of \cite{G3} described, roughly speaking, only half of emerging eigenvalues.}

{The bifurcations of the thresholds in the essential spectrum can be also studied for
perturbations of non-self-adjoint operators provided the spectral structure of the limiting operator is known in sufficiently great details. As examples, we mention works \cite{CP}, \cite{V}, where an evolutionary nonlinear Schr\"odinger equation was considered with both linear and nonlinear perturbations. The linearization of this equation on solitary wave solutions gave rise to a spectral problem for a linear non-self-adjoint operator. The essential spectral of such operator consists of two real semi-axes; the main results of \cite{CP}, \cite{V} provided conditions, under which the end-points of the essential spectrum bifurcated into eigenvalues. If the latter existed, their two-terms asymptotic expansions were calculated.}

An important feature of the eigenvalues emergence is that usually the total multiplicity of the emerging eigenvalues does not exceed the multiplicity of edge in the essential spectrum from which they emerge. The multiplicity of the edge is to  be treated in the sense of some appropriate generalized eigenfunctions. However, there were found examples, when this commonly believed rule failed. The earliest work on this subject we know is paper \cite{GH}, where the Schr\"odinger operator on $\mathds{R}^3$ perturbed by a small localized potential was considered. It was found that in certain cases, an $n$-multiple bottom of the essential spectrum can generate $n$ eigenvalues and $n$ anti-bound states or $2n$ resonances. In \cite{PMA14}, a similar phenomenon was found for  the Dirichlet Laplacian in a pair of three-dimensional layers coupled by a window, when the perturbation was a small variation of the window shape. Very recently we succeeded to find  an even more impressive example of infinitely many eigenvalues and/or resonances emerging from the bottom of   an   essential spectrum. This was done in papers \cite{JPA19}, \cite{AML20}, where we considered  {a} 
one-dimensional Schr\"odinger operator on the axis with two complex localized potentials,
 {the supports of which}  
 were separated by a large distance. It was found that as this distance increases, more and more resonances and eigenvalues appear in the vicinity of the bottom of the essential spectrum, while the multiplicity of this bottom is at most one. The location and asymptotic behaviour of these emerging eigenvalues and resonances were analyzed in details.

Emerging eigenvalues were also studied not only in the context of classical eigenvalue problems, but also for more complicated operator pencils. In \cite{ESAIM20}, there was considered a special quadratic operator pencil on the line with a special small periodic  $\PT$-symmetric perturbation. The structure of the gaps in the essential spectrum and complex eigenvalues in the vicinities of the edges of these gaps were analyzed in great details. In \cite{SAM17}, there was considered a similar quadratic operator pencil with a special small localized $\PT$-symmetric potential and there were studied eigenvalues emerging from thresholds in the essential spectrum. Sufficient existence conditions were established and the leading terms of the asymptotic expansions of the emerging eigenvalues were obtained.

\paragraph{Model and results.}  In the present paper we carry out a rigorous analysis of bifurcations of isolated eigenvalues and resonances from the essential spectrum of a   multi-dimensional operator under a small localized general abstract perturbation. Namely, we consider a self-adjoint operator of the form
\begin{equation*}
  \Op=\Op' -\frac{\p^2\ }{\p x_d^2} \quad\text{on}\quad\om\times\mathds{R}
\end{equation*}
subject to the Dirichlet or Robin condition, where $\om\subseteq\mathds{R}^{d-1}$ is some domain, which can be both bounded and unbounded and $\Op'$ is a self-adjoint second order differential operator on $\om$
subject to the same boundary condition as $\Op$. We assume that the spectrum of the operator $\Op'$ contains several discrete eigenvalues $\L_1\leqslant \L_2\leqslant\ldots \leqslant \L_m$ below the essential spectrum. Then the essential spectrum of the operator $\Op$  is the half-line $[\L_1,+\infty)$ and the mentioned eigenvalues become thresholds in this essential spectrum. We add a small localized perturbation to the operator $\Op$. This perturbation reads as $\e\cL(\e)$, where $\e$ is a small positive parameter and $\cL(\e)$ is an abstract not necessarily symmetric operator acting from a weighted Sobolev space $W_2^2(\Om,e^{-{\vt}|x_d|}dx)$ into a weighted Lebesgue space $L_2(\Om,e^{{\vt}|x_d|}dx)$. Exact definitions of these spaces will be given in the next section and now we just say that these weights and the operator $\cL(\e)$ are designed so that this operator maps exponentially growing functions into exponentially decaying ones. The latter fact is exactly how we understand the localization of this operator in a generalized sense.

The main result of our paper describes how the thresholds $\L_j$ in the essential spectrum of the operator $\Op$ bifurcate under the presence of the perturbation $\e\cL(\e)$. We show that if the bottom of the essential spectrum $\L_1$ is an $m$-multiple eigenvalue of the operator $\Op'$, then there can be at most $m$ eigenvalues and resonances of the operator $\Op$ in the vicinity of $\L_1$ for sufficiently small $\e$. The vicinity of an internal threshold
$\L_j>\L_1$ in the essential spectrum being an  $m$-multiple eigenvalue of the operator $\Op'$ can contain at most $2m$ eigenvalues and resonances of the operator $\Op$. The eigenvalues and resonances are identified via an analysis of the poles of an appropriate 
 {meromorphic}
 continuation of the resolvent in the vicinity of each threshold $\L_j$, $j\geqslant 1$.  Each such pole generates either an eigenvalue or a resonance, and we provide simple sufficient conditions allowing one to identify whether a considered pole is an eigenvalue or a resonance. We also construct two-terms asymptotic expansions for the emerging eigenvalues and resonances.

\paragraph{Applications to specific models.}  While our result is rather general
and applies to a broad range of physical models, where elliptic operators play the prominent role, we will exemplify applications of the work using some
particular physical models. Namely,
we discuss various examples of the unperturbed operator $\Op$ and of the perturbation $\cL(\e)$. Then we consider   models of two- and three-dimensional waveguides  and   models of two- and three-dimensional quantum oscillators. As a perturbation, we choose a small localized complex potential. Such choice is motivated by physical models of optical waveguides  filled with a homogeneous medium, when the refractive index of the waveguide is locally modulated. This creates a small, generally, non-Hermitian  perturbation   in the form of an   effective complex-valued optical potential. In particular, this potential can be  $\PT$-symmetric.
Another physical model motivating the above examples is
a
two-dimensional Bose-Einstein condensate trapped in a harmonic potential in one dimension and without any trapping in the second dimension. The nonlinear interactions between particles of the condensate are assumed to be negligible such that its evolution can be described by the   linear Schr\"odinger operator. As a perturbation,
a localized non-Hermitian defect such as a localized dissipation serves.  The similar approach can be applied to a  three-dimensional condensate, where a localized perturbation can trigger formation of fully localized structures with internal vorticity.

For such examples we show that given an internal threshold of a multiplicity $n$,  by tuning appropriately the perturbing potential, we can make the threshold to bifurcate into $n$ pairs of complex-conjugated eigenvalues. This example demonstrates that the total multiplicity of the emerging eigenvalues can exceed the multiplicity of the internal threshold.

\paragraph{Organization of the paper.} The rest of this paper is organized as follows.  In  Section~\ref{sec:problem} we elaborate rigorous mathematical formulation of the problem, and then present and discuss the main results which are formulated in several theorems. Section~\ref{sec:examples} is dedicated to examples, including a case study of optical bound states in the continuum emerging under a small $\PT$-symmetric perturbation. Sections~\ref{sec:proof1} and \ref{sec:proof2} contain the proofs of theorems.

\section{Problem and  results}
\label{sec:problem}
\subsection{Problem}

Let $x'=(x_1,\ldots,x_{d-1})$, $x=(x',x_d)$ be Cartesian coordinates in $\mathds{R}^{d-1}$ and $\mathds{R}^d$, respectively, where $d\geqslant 2$, and $\om\subseteq\mathds{R}^{d-1}$ be an arbitrary domain. The domain $\om$ can be bounded or unbounded, the case $\om=\mathds{R}^{d-1}$ is also possible. If the boundary of the domain $\om$ is non-empty, we assume that $\p\om\in C^2$.
We let $\Om:=\om\times\mathds{R}$ and we suppose that the domain $\om$   is such that
\begin{equation*}
\|u\|_{L_2(\p\om)}\leqslant C\|u\|_{W_2^1(\om)}
\end{equation*}
for all $u\in W_2^1(\om)$ with a constant $C$ independent of $u$. This inequality implies that
\begin{equation*}
\|u\|_{L_2(\p\Om)}\leqslant C\|u\|_{W_2^1(\Om)}
\end{equation*}
for all $u\in W_2^1(\Om)$ with a constant $C$ independent of $u$. This means that on the boundary of the domain $\Om$, the traces of the functions in $W_2^1(\Om)$ are well-defined and the trace operator is bounded. This fact is employed below in definitions of various operators and sesquilinear forms without explicit mentioning.

By $A_{ij}=A_{ij}(x')$, $A_j=A_j(x')$, $A_0=A_0(x')$, $i,j=1,\ldots,d-1$, we denote real functions defined on $\overline{\om}$ and with the following smoothness: $A_{ij}, A_j\in C^1(\overline{\om})$, $A_0\in C(\overline{\om})$. The functions $A_{ij}$ satisfy the usual uniform ellipticity condition, that is, $A_{ij}=A_{ji}$ and
\begin{equation*}
\sum\limits_{i,j=1}^{d-1} A_{ij}{(x')} \xi_i\overline{\xi_j}\geqslant c_0\sum\limits_{i=1}^{d-1} |\xi_i|^2\quad\text{for all}\quad x'\in\overline{\om}, \quad \xi_i\in\mathds{C},
\end{equation*}
where $c_0>0$ is a positive constant independent of {$x'$} and $\xi_i$. The functions $A_{ij}$ and $A_j$ are assumed to be uniformly bounded on $\overline{\om}$, while for $A_0$ only an uniform lower bound is supposed. By $\iu$ we denote the imaginary unit.

In terms of the introduced functions we define an operator
\begin{equation}\label{2.0}
\Op=-\sum\limits_{i,j=1}^{d-1}\frac{\p\ }{\p x_i} A_{ij} \frac{\p\ }{\p x_j}-\frac{\p^2\ }{\p x_d^2} + \iu \sum\limits_{j=1}^{d-1} \left(A_j\frac{\p\ }{\p x_j} +\frac{\p\ }{\p x_j}A_j \right)+A_0\quad\text{in}\quad  \Om
\end{equation}
 subject to   the Dirichlet condition or Robin condition:
\begin{equation}\label{2.3}
u=0\quad\text{on}\quad\p\Om\qquad\text{or}\qquad \frac{\p u}{\p \boldsymbol{\nu}}-a u=0\quad\text{on}\quad\p\Om.
\end{equation}
In the case of Robin condition, the conormal derivative is defined as
\begin{equation*}
\frac{\p u}{\p \boldsymbol{\nu}}:=\sum\limits_{i,j=1}^{d-1} A_{ij}\nu_i \frac{\p u}{\p x_j}-\iu\sum\limits_{j=1}^{d-1} A_j \nu_j {u} + \nu_d \frac{\p u}{\p x_d},
\end{equation*}
where $\nu=(\nu_1,\ldots,\nu_d)$ is the unit outward normal to $\p\Om$ and $a=a(x')$ is a real function defined on $\p\Om$. We assume that $a\in C(\p\Om)$ and that this function is uniformly bounded on $\p\Om$. We define a chosen boundary operator in (\ref{2.3}) by $\cB$, that is, $\cB u=u$ or $\cB u=\frac{\p u}{\p \boldsymbol{\nu}}-a u$.

Rigorously we introduce the operator $\Op$ as follows.
In the space $L_2(\Om)$ we define a sesquilinear form
\begin{align*}
\hf(u,v):=&\sum\limits_{i,j=1}^{d-1}\left(A_{ij}\frac{\p u}{\p x_j}, \frac{\p v}{\p x_i}\right)_{L_2(\Om)} + \left(\frac{\p u}{\p x_d}, \frac{\p v}{\p x_d} \right)_{L_2(\Om)}+ \iu\sum\limits_{j=1}^{d-1} \left(A_j\frac{\p u}{\p x_j}, v\right)_{L_2(\Om)}
\\
&
-\iu \sum\limits_{j=1}^{d-1} \left(u, A_j\frac{\p v}{\p x_j}\right)_{L_2(\Om)}
+
 (A_0 u,v)_{L_2(\Om)}
\end{align*}
on the domain $\Dom(\hf):=\Ho^1(\Om)\cap L_2(\Om,(1+|A_0|)dx)$ if the Dirichlet condition is chosen in (\ref{2.3}),
and
\begin{align*}
\hf(u,v):=&\sum\limits_{i,j=1}^{d-1}\left(A_{ij}\frac{\p u}{\p x_j}, \frac{\p v}{\p x_i}\right)_{L_2(\Om)} + \left(\frac{\p u}{\p x_d}, \frac{\p v}{\p x_d} \right)_{L_2(\Om)}+\iu \sum\limits_{j=1}^{d-1} \left(A_j\frac{\p u}{\p x_j}, v\right)_{L_2(\Om)}
\\
&
-\iu \sum\limits_{j=1}^{d-1} \left(u, A_j \frac{\p v}{\p x_j}\right)_{L_2(\Om)}
  +
 (A_0 u,v)_{L_2(\Om)}
  - (a u,v)_{L_2(\p\Om)}
\end{align*}
on the domain $\Dom(\hf):=W_2^1(\Om)\cap L_2(\Om,(1+|A_0|)dx)$. Here $\Ho^1(\Om)$ is a subspace of the space $W_2^1(\Om)$ consisting of the functions with a zero trace on $\p\Om$.  Given a positive function {$\phi$} on $\Om$, by $L_2(\Om, \phi dx)$ we denote a weighted space formed by the functions in $L_{2,loc}(\Om)$ with a finite norm $\|\cdot\|_{L_2(\Om, \phi dx)}$ defined as
\begin{equation*}
\|u\|_{L_2(\Om, \phi\,dx)}^2=\int\limits_{\Om}
  |u(x)|^2 \phi\, dx.
\end{equation*}

Thanks to the above assumptions on the functions $A_{ij}$, $A_j$, $A_0$ and $a$, the form $\hf$ is closed, symmetric and lower-semibounded. The self-adjoint operator
in $L_2(\Om)$ associated with this form is exactly the operator $\Op$.

We introduce one more weighted space  $W_2^2(\Om,e^{-{\vt}|x_d|}dx)$  as a subspace of $W_{2,loc}^2(\Om)$  formed by the functions with finite norms $\|\cdot\|_{W_2^2(\Om,e^{-{\vt}|x_d|}dx)}$, is defined as follows:
\begin{equation*}
\|u\|_{W_2^2(\Om,e^{-{\vt}|x_d|}dx)}^2=\int\limits_{\Om}
\sum\limits_{\substack{\a\in\mathds{Z}_+^2\\|\a|\leqslant 2}} |\p^\a u(x)|^2 e^{-{\vt}|x_d|}dx.
\end{equation*}
Here {$\vt>0$} is some fixed constant.
By $\e$ we denote a small positive parameter and the symbols $\cL_1$,   $\cL_2$, $\cL_3=\cL_3(\e)$ stand for operators mapping the {space} 
 $W_2^2(\Om,e^{-{\vt}|x_d|}dx)$ into $L_2(\Om,e^{{\vt}|x_d|}dx)$. These operators are assumed to be bounded; the operator $\cL_3$ is bounded uniformly in $\e$. We stress that the operators  $\cL_1$, $\cL_2$, $\cL_3$ are not {supposed to be necessarily} 
 symmetric.

The main object of our study is a perturbed operator
\begin{equation*}
\Op_\e=\Op+\e\cL(\e),\qquad \cL(\e):=\cL_1+\e\cL_2+\e^2\cL_3(\e)
\end{equation*}
on $\Dom(\Op)$. The operator $\Op_\e$ is well-defined since
\begin{equation*}
\Dom(\Op)\subseteq W_2^2(\Om) \subseteq W_2^2(\Om,e^{-{\vt}|x_d|}dx).
\end{equation*}
Moreover, it is clear that the operator $\cL(\e)$ is relatively bounded with respect to the operator $\Op$ and this is why, for sufficiently small $\e$, the operator $\Op_\e$ is closed.

Our main {aim} is to study the behaviour of the eigenvalues of the operator $\Op_\e$ emerging from certain internal points in its essential spectrum. We denote the latter by $\essspec(\cdot)$ and define it in terms of a characteristic sequences. Namely, a point $\l$ belongs to an essential spectrum $\essspec(\cA)$ of some operator $\cA$ if there exists a bounded noncompact sequence $u_d\in\Dom(\cA)$ such that
\begin{equation*}
\inf\limits_{d}\|u_d\|>0\qquad\text{and}\qquad (\cA-\l)u_d\to0,\quad d\to\infty.
\end{equation*}

In order to describe the essential spectrum $\essspec(\Op_\e)$, we introduce  two auxiliary operators {$\Op'$} and $\Op_0$. The former is a self-adjoint operator in $L_2(\mathds{R}^{d-1})$ associated with a lower-semibounded symmetric sesquilinear
form
\begin{align*}
\hf'(u,v):=&\sum\limits_{i,j=1}^{d-1}\left(A_{ij}\frac{\p u}{\p x_j}, \frac{\p v}{\p x_i}\right)_{L_2(\om)}
+ \iu\sum\limits_{j=1}^{d-1}  \left(A_j\frac{\p u}{\p x_j}, v\right)_{L_2(\om)}
\\
&-\iu\sum\limits_{j=1}^{d-1} \left(u, A_j\frac{\p v}{\p x_j}\right)_{L_2(\om)}
  +
 (A_0 u,v)_{L_2(\om)}
\end{align*}
on the domain $\Dom(\hf'):=\Ho^1(\om)\cap L_2(\om,(1+|A_0|)dx')$ if the Dirichlet condition is chosen in (\ref{2.3}) and
\begin{align*}
\hf'(u,v):=&\sum\limits_{i,j=1}^{d-1}\left(A_{ij}\frac{\p u}{\p x_j}, \frac{\p v}{\p x_i}\right)_{L_2(\om)}
+   \iu\sum\limits_{j=1}^{d-1}  \left(A_j\frac{\p u}{\p x_j}, v\right)_{L_2(\om)}
\\
&
  - \iu\sum\limits_{j=1}^{d-1}\left(u, A_j \frac{\p v}{\p x_j}\right)_{L_2(\om)}
+
 (A_0 u,v)_{L_2(\om)}-(au,v)_{L_2(\p\om)}
\end{align*}
on the domain $\Dom(\hf'):=W_2^1(\om)\cap L_2(\om,(1+|A_0|)dx')$ if the Robin condition is chosen in (\ref{2.3}). This is the operator
\begin{equation*}
\Op'=-\sum\limits_{i,j=1}^{d-1}\frac{\p\ }{\p x_i} A_{ij} \frac{\p\ }{\p x_j} + \iu \sum\limits_{j=1}^{d-1} \left(A_j\frac{\p\ }{\p x_j} +\frac{\p\ }{\p x_j}A_j \right)+A_0\quad\text{in}\quad  \om
\end{equation*}
 subject to the Dirichlet condition or Robin condition:
\begin{equation*}
u=0\quad\text{on}\quad\p\om\qquad\text{or}\qquad \frac{\p u}{\p \boldsymbol{\nu}'}=a u\quad\text{on}\quad\p\Om,\qquad
\frac{\p u}{\p \boldsymbol{\nu}'}:=\sum\limits_{i,j=1}^{d-1} A_{ij}\nu_i \frac{\p u}{\p x_j}-\iu\sum\limits_{j=1}^{d-1} A_j \nu_j.
\end{equation*}

The operator $\Op_0$ is a one-dimensional Schr\"odinger operator
\begin{equation*}
\Op_0:=-\frac{d^2\ }{dx_d^2}
\end{equation*}
in $L_2(\mathds{R})$ on the domain $W_2^2(\mathds{R})$. Its spectrum is pure essential and coincides with $[0,+\infty)$. We assume that there exists a constant $c_0$ such that the spectrum of the  operator $\Op'$ below this constant consists of finitely many discrete eigenvalues, which we denote by $\L_j$ and we {arrange} 
them in an ascending order counting multiplicities:
\begin{equation*}
\L_1\leqslant \L_2\leqslant\ldots\leqslant\L_m<c_0.
\end{equation*}
The associated orthonormalized in $L_2(\mathds{R}^{d-1})$ eigenfunctions are denoted by $\psi_j=\psi_j(x')$, $j=1,\ldots,m$.

The  essential spectrum of the operator $\Op_\e$ is described in the following lemma.

\begin{lemma}\label{lm2.1}
The essential spectrum of the operator $\Op_\e$ coincides with that of the operator $\Op$ for all sufficiently small $\e$ and   is given by the identity:
\begin{equation*}
\essspec(\Op_\e)=\essspec(\Op)=\spec(\Op)=[\L_1,+\infty).
\end{equation*}
where $\spec(\cdot)$ denotes a spectrum of an operator.
\end{lemma}

According {to} this lemma, the points $\L_j$, $j=1,\ldots,m$, belong to the essential spectrum of the operator $\Op_\e$. The point $\L_1$ is the bottom of such spectrum, while other points $\L_j$ are internal thresholds.

\subsection{Main results}

Our   results describe 
 {a meromorphic} continuation of the resolvent of the operator $\Op_\e$ in the vicinity of the points $\L_j$ as well as eigenvalues and resonances emerging from  these points due to the presence of the perturbation $\e\cL(\e)$.
Before presenting our main results, we introduce some auxiliary constants and notations.

By $B_\d$ we denote a ball of radius $\d$ centered at the origin in the complex plane. We fix $p\in\{1,\ldots,m\}$ and assume that $\L_p=\ldots=\L_{p+n-1}$ is an  $n$-multiple eigenvalue of the operator $\Op'$, where $n\geqslant1$. Then we consider a new complex parameter $k$ ranging in a small neighbourhood of the origin and we introduce auxiliary functions:
\begin{align*}
&K_j(k):=-\iu\sqrt{\L_p-\L_j-k^2} \qquad\text{as}\quad j<p,
\\
&K_j(k):=k\hphantom{1\sqrt{\L_p-\L_j-k^2}}  \qquad\text{as}\quad j=p,\ldots,p+n-1,
\\
&
K_j(k):=\sqrt{\L_j-\L_p+k^2}\hphantom{-\iu} \qquad\text{as}\quad j\geqslant p+n.
\end{align*}
Hereinafter the branch of the square root is fixed by the condition $\sqrt{1}=1$ {with the branch cut along the negative real semi-axis}. Given $R>0$, we let $\Om_\pm^R:=\Om\cap\{x: \pm x_d>R\}$.

Now we are in position to formulate our first main result.

\begin{theorem}\label{thAnCo} Fix $p\in\{1,\ldots, m\}$,
 and $\tau\in\{-1,+1\}$ and let $\L_p=\ldots=\L_{p+n-1}$ be an  $n$-multiple eigenvalue of the operator $\Op'$, where $n\geqslant1$.
For all sufficiently small $\e$, the resolvent $(\Op_\e{-}\L_p+k^2)^{-1}$ admits {a meromorphic} 
continuation with respect to a complex parameter $k$ ranging in a sufficiently small neighbourhood of the origin. Namely, there exists a bounded operator \begin{equation*}
\cR_{\e,\tau}(k):\, L_2(\Om,e^{{\vt}|x_d|}\,dx)\to W_2^2(\Om,e^{-{\vt}|x_d|}\,dx)
\end{equation*}
meromorphic with respect to complex $k\in B_\d$ for a sufficiently small fixed $\d$  independent of $\e$.
If $p=1$, then the operator $\cR_{\e,\tau}(k)$ is independent of the choice of $\tau$ and for $\RE k>0$, this operator coincides with the resolvent $(\Op_\e- \L_1+k^2)^{-1}$ restricted on $L_2(\Om,e^{{\vt}|x_d|}\,dx)$. If $p>1$, then  the operator $\cR_{\e,\tau}(k)$ does depend on the choice of $\tau$
and for
$\RE k>0$ and $\tau\IM k^2<0$, this operator   coincides with the resolvent $(\Op_\e- \L_p+k^2)^{-1}$ restricted on $L_2(\Om,e^{{\vt}|x_d|}\,dx)$.

For all $f\in L_2(\Om,e^{{\vt}|x_d|}\,dx)$, the function $u_\e:=\cR_{\e,\tau}(k)f$ solves the boundary value problem
\begin{equation}\label{2.5}
 \begin{gathered}
\left(
 -\sum\limits_{i,j=1}^{d-1}\frac{\p\ }{\p x_i} A_{ij} \frac{\p\ }{\p x_j} -\frac{\p^2\ }{\p x_d^2}+ \iu \sum\limits_{j=1}^{d-1} \left(A_j\frac{\p\ }{\p x_j} +\frac{\p\ }{\p x_j}A_j \right)+A_0 + \e\cL(\e) -\L_p+k^2
 \right)u_\e=f\quad\text{in}\quad\Om,
 \\
\cB u_\e=0 \quad\text{on}\quad\p\Om,
\end{gathered}
\end{equation}
 and for sufficiently large $x_d$ it can be represented as follows:
\begin{equation}\label{2.6}
u_\e(x,k)=\sum\limits_{j=1}^{m} u_{\e,j}^\pm(x_d,k)\psi_j(x') +u_{\e,\bot}^\pm(x,k),\qquad \pm x_d>R,
\end{equation}
where $R$ is some fixed number, $u_{\e,j}^\pm\in L_2(I_\pm,e^{{\mp\vt} x_d}dx_d)$ are some meromorphic in $k\in B_\d$
functions, $I_+:=(R,+\infty)$, $I_-:=(-\infty,-R)$, possessing the asymptotic behavior
\begin{equation}\label{2.6a}
\begin{aligned}
&u_{\e,j}^\pm(x_d)= e^{-\tau K_j(k)|x_d|}\big(C_{\e,j}^\pm(k)+O({e}^{-{\tilde{\vt}}|x_d|})\big),\quad && |x_d|\to\infty,\qquad j=1,\ldots,p-n+1,
\\
&u_{\e,j}^\pm(x_d)=e^{- K_j(k)|x_d|}\big(C_{\e,j}^\pm(k)+O({e}^{-{\tilde{\vt}}|x_d|})\big), && |x_d|\to\infty,\qquad j=p,\ldots,m,
\end{aligned}
\end{equation}
$C_{\e,j}^\pm(k)$ are some meromorphic in $k\in B_\d$ functions, $0<{\tilde{\vt}<\vt}$ is some fixed constant independent of $k$ and $x$,
and $u_{\e,{\bot}}^\pm\in W_2^2(\Om_R^\pm)$ are some functions meromorphic in $k\in B_\d$ and obeying the identities
\begin{equation}\label{2.6b}
(u_{\e,\bot}^\pm(\cdot,x_d),\psi_j)_{L_2({\om})}=0
\end{equation}
for almost each $x_d\in I_\pm$ and for each $j=1,\ldots,m$.

If $k_\e\in B_\d$ is a pole of the operator $\cR_{\e,\tau}(k)$, for $k=k_\e$, problem (\ref{2.5}) with $f=0$ has a non-trivial solution $\psi_\e$ in $W_{2,loc}^2(\Om)$, which
satisfies a representation similar to (\ref{2.6}):
\begin{equation}\label{2.8}
\psi_\e(x)=
\sum\limits_{j=1}^{m} \phi_{\e,j}^\pm(x_d,k)\psi_j(x') +\psi_{\e,\bot}^\pm(x,k),\qquad \pm x_d>R,
\end{equation}
where $R$ is some fixed number,
$\phi_{\e,j}^\pm\in L_2(I_\pm,e^{{\mp\vt} x_d}dx_d)$ are
functions with the asymptotic behavior
\begin{equation}\label{2.8a}
\begin{aligned}
&\phi_{\e,j}^\pm(x_d)= e^{-\tau K_j(k)|x_d|}\big(c_{\e,j}^\pm(k)+O({e}^{-{\tilde{\vt}}|x_d|})\big),\quad && |x_d|\to\infty,\qquad j=1,\ldots,p-n+1,
\\
&\phi_{\e,j}^\pm(x_d)=e^{- K_j(k)|x_d|}\big(c_{\e,j}^\pm(k)+O({e}^{-{\tilde{\vt}}|x_d|})\big), && |x_d|\to\infty,\qquad j=p,\ldots,m,
\end{aligned}
\end{equation}
$c_{\e,j}^\pm(k)$ are some constants,
and $\psi_\e^\pm\in W_2^2(\Om_R^\pm)$ are some functions  obeying the identities
\begin{equation}\label{2.8b}
(\psi_{\e,\bot}^\pm(\cdot,x_d),\psi_j)_{L_2({\om})}=0
\end{equation}
for almost each $x_d\in I_\pm$ and for each $j=1,\ldots,m$.
\end{theorem}

We define a subspace $L^\bot$ in $L_2(\Om)$ as a set of functions $v\in L_2(\Om)$ such that
\begin{equation*}
(v(\cdot,x_d),\psi_j)_{L_2(\om)}=0
\end{equation*}
for almost each $x_d\in\mathds{R}$ and for all $j=1,\ldots,m$.
The space $L^\bot$ is a Hilbert one.
By $\Op^\bot$ we denote the restriction of the operator $\Op$ on
$\Dom(\Op)\cap L^\bot$. The following lemma will be proved in Section~\ref{ssAuLms}.

\begin{lemma}\label{lmLiRes}
The space $L^\bot$ is invariant for the operator $\Op^\bot$, that is, this operator maps $\Dom(\Op)\cap L^\bot$ into $L^\bot$. This is an  unbounded self-adjoint operator in $L^\bot$  and its spectrum is located in $[c_0, +\infty)$.
\end{lemma}

The above lemma means that the resolvent $(\Op^\bot-\L_p)^{-1}$ is well-defined for all $p=1,\ldots,m$ as an operator from $L_2(L^\bot)$ into $\Dom(\Op)\cap L^\bot$.  As above, we fix $p\in\{1,\ldots,m\}$ and assume that $\L_p=\ldots=\L_{p+n-1}$, where $n\geqslant1$, and in terms of the latter resolvent, we introduce {an} auxiliary {operator }
mapping $L_2(\Om,e^{{\vt}|x_d|}dx)$ into $W_2^2(\Om,e^{{-\vt}|x_d|}dx)$:
\begin{align}
&
\begin{aligned}
(\cG_{p,\tau} f)(x):=&
\sum\limits_{j=1}^{p-1}  \frac{\psi_j(x') }{2\tau K_j(0)} \int\limits_{\Om} e^{-\tau K_j(0)|x_d-y_d|}\overline{\psi_j(y')} f(y)\,dy
\\
&-\frac{1}{2}\sum\limits_{j=p}^{p+n-1}  \psi_j(x') \int\limits_{\Om}|x_d-y_d|\overline{\psi_j(y')}f(y)\di y
\\
&+ \sum\limits_{j=p+n}^{m}  \frac{\psi_j(x')}{2 K_j(0)} \int\limits_{\Om} e^{-K_j(0)|x_d-y_d|}\overline{\psi_j(y')} f(y)\,dy+((\Op^\bot-\L_p)^{-1}f^\bot)(x),
\end{aligned}
\label{3.9}
\\
\label{3.11}
&f^\bot(x):=f(x)-\sum\limits_{j=1}^{m} f_j(x_d) \psi_j(x').
\end{align}
As above, here $\tau\in\{-1,+1\}$. In the case $p=1$, the first sum in the above definition is missing and the operator $\cG_{p,\tau}$ becomes independent of the choice of $\tau$.

We define the matrix $\rM_1$ with entries
\begin{equation}\label{2.9}
M_1^{ij}:=-\frac{1}{2}\int\limits_{\Om} \overline{\psi_{i+p-1}} \cL_1\psi_{j+p-1}\di x,
\qquad i,j=1,\ldots, n,
\end{equation}
where $i$ counts the rows and $j$ does the columns in the matrix $\rM_1$; {
 since the operator $\cL_1$ acts from $W_2^2(\Om,e^{-\vt|x_d|}dx)$ into $L_2(\Om,e^{\vt|x_d|}dx)$, the functions $\cL_1\psi_{j+p-1}$ belong to $L_2(\Om,e^{\vt|x_d|}dx)$ and this obviously ensures the convergence of the integrals in the above identity.}

By $\mu_i$, $i=1,\ldots,N$, we denote different eigenvalues of the matrix $\rM_1$ of multiplicities $q_1, \ldots, q_N$. It is clear that $N\leqslant n$ and $q_1+\ldots+q_N=n$.

\begin{theorem}\label{thEmer1}
Fix $p\in\{1,\ldots,m\}$, $\tau\in\{-1,+1\}$ and let $\L_p=\ldots=\L_{p+n-1}$ be an  $n$-multiple eigenvalue of the operator $\Op'$, where $n\geqslant1$.  There are exactly $N$ poles{,} counting their orders{,} of the operator $\cR_{\e,\tau}(k)$ converging to zero as $\e\to+0$. These poles, denoted by $k_{ij}(\e)$, have the asymptotic behavior  
\begin{equation}\label{2.10}
k_{ij}(\e)=\e\mu_i + O\big(\e^{1+\frac{1}{q_i}}\big),\qquad i=1,\ldots,N,\quad j=1,\ldots,q_i.
\end{equation}
\end{theorem}

Asymptotic expansion (\ref{2.10}) for the poles $k_{ij}$ can be specified in more details and this will be done in terms of  one more matrix   $\rM_{2,\tau}$ with entries
\begin{equation}\label{2.11}
M_{2,\tau}^{ij}:=\frac{1}{2}\int\limits_{\Om} \overline{\psi_{i+p-1}} (\cL_2-\cL_1 \cG_{p,\tau} \cL_1)\psi_{j+p-1}\di x
\qquad i,j=1,\ldots, n,
\end{equation}
where $i$ counts the rows and $j$ does the columns in the matrix $\rM_{2,\tau}$. {Due to the definition of  the operators $\cL_1$, $\cL_2$ , the second term in the integrand in (\ref{2.11}) belongs to $L_2(\mathds{R},e^{\vt|x_d|}dx_d)$ and this ensures the convergence of the integral.}
We denote
\begin{equation}\label{2.12}
Q_{i,\tau}(z):=\frac{\p\ }{\p\e}\det\big(z\rE-\rM_1+\e\rM_{2,\tau}\big)\bigg|_{\e=0}.
\end{equation}
We stress that if $\L_p=\L_1$, the matrix $M_{2,\tau}$ and the function $Q_{i,\tau}$ become independent of $\tau$.

\begin{theorem}\label{thEmer2}
Under the assumptions of Theorem~\ref{thEmer1}, we fix $i\in\{1,\ldots,N\}$. If $Q_{i,\tau}(z)$  vanishes identically, then
\begin{equation}\label{2.13}
k_{ij}(\e)=\e\mu_i + O\big(\e^{1+\frac{2}{q_i}}\big),\qquad i=1,\ldots,N,\quad j=1,\ldots,q_i.
\end{equation}

If $Q_{i,\tau}$ is not identically zero, then there exists a fixed non-negative integer $r_{i,\tau}<q_i$ such that
\begin{equation}\label{2.16}
\g_{i,\tau}:=\frac{r_{i,\tau}!}{\prod\limits_{\substack{j=1
\\
j\ne i}}^{N}(\mu_i-\mu_j)^{q_j}} \frac{d^{r_{i,\tau}} Q_{i,\tau}}{dz^{r_{i,\tau}}}(\mu_i)\ne0.
\end{equation}
If $2r_{i,\tau}\geqslant q_i$, then
\begin{equation}\label{2.17}
k_{ij}(\e)=\e\mu_i + O\big(\e^{1+\frac{1}{r_{i,\tau}}}\big),\qquad i=1,\ldots,N,\quad j=1,\ldots,q_i.
\end{equation}
If $2r_{i,\tau}\leqslant q_i-1$, then exactly $r_{i,\tau}$ poles $k_{ij}$, $j=1,\ldots,r_{i,\tau}$ have the asymptotic behavior
\begin{equation}\label{2.18}
k_{ij}(\e)=\e\mu_i + O\big(\e^{1+\frac{1}{r_{i,\tau}}}\big),\qquad i=1,\ldots,N,\quad j=1,\ldots,r_{i,\tau},
\end{equation}
while other poles $k_{ij}$, $j=r_{i,\tau}+1,\ldots,q_i$, have the asymptotic behavior
\begin{equation}\label{2.19}
k_{ij}(\e)=\e\mu_i +\e^{1+\frac{1}{q_i-r_{i,\tau}}} (-\g_{i,\tau})^{\frac{1}{q_i-r_{i,\tau}}}e^{\frac{2\pi\iu}{q_i-r_{i,\tau}}(j-r_{i,\tau})} +O\big(\e^{1+\frac{2}{q_i-r_{i,\tau}}}\big),
\end{equation}
where the branch of the fractional power $z^{\frac{1}{q_i-r_{i,\tau}}}$ is fixed by the condition $1^{\frac{1}{q_i-r_{i,\tau}}}=1$ {with the branch cut along the negative real semi-axis.}
\end{theorem}

We give some definitions before we formulate our next result.
A pole $k_\e\in B_\d$ of an operator $\cR_{\e,\tau}(k)$ corresponds to an eigenvalue $\L_p-k_\e^2$ of an operator $\Op_\e$ if an associated  nontrivial solution $\psi_\e$ to (\ref{2.5}), (\ref{2.8}) belongs to $W_2^2(\Om)$. 
{Otherwise} it corresponds to a resonance $\L_p-k_\e^2$.

Our further results provide  conditions allowing to determine whether resonances or eigenvalues are associated with the poles described in two previous theorems. We first present the main result on the poles emerging from the bottom of the essential spectrum.

\begin{theorem}\label{thEmBot} Let $p=1$ and make the assumptions of Theorems~\ref{thEmer1},~\ref{thEmer2}. If $\RE\mu_i>0$,
then the   poles $k_{ij}$, $j=1,\ldots,q_i$ correspond to the eigenvalues $\l_{ij}(\e)=\L_p-k_{ij}^2(\e)$ with the asymptotic behavior
\begin{equation}\label{6.4}
\l_{ij}(\e)=\L_p-\e^2\mu_i^2+O\big(\e^{2+\frac{1}{\a_i}}\big)
\end{equation}
with $j=1,\ldots,q_i$, where
\begin{equation}\label{6.5}
\a_i:=\left\{
\begin{aligned}
&\frac{q_i}{2}\qquad\hphantom{\tau}\text{if}
\quad Q_{i,\tau}\ \text{vanishes identically},
\\
&\,r_{i,\tau}\qquad\text{if}\quad 2r_{i,\tau}\geqslant q_i.
\end{aligned}
\right.
\end{equation}
If $2r_{i,\tau}\leqslant q_i-1$, then the eigenvalues $\l_{ij}$ still have asymptotic behavior (\ref{6.4}) with $\a_i=r_{i,\tau}$ for $j=1,\ldots,r_{i,\tau}$, while the asymptotic behaviors for the other eigenvalues read as
\begin{equation}\label{6.6}
\l_{ij}(\e)=\L_p-\e^2\mu_i^2-2\e^{2+\frac{1}{q_i-r_{i,\tau}}}(-\g_{i,\tau})^{\frac{1}{q_i-r_{i,\tau}}}
e^{\frac{2\pi\iu}{q_i-r_{i,\tau}}(j-r_{i,\tau})}
+O(\e^{2+\frac{2}{q_i-r_{i,\tau}}}),\qquad j=r_{i,\tau}+1,\ldots,q_i.
\end{equation}

If $\RE\mu_i<0$, then the   poles $k_{ij}$, $j=1,\ldots,q_i$ correspond to the resonances $\l_{ij}(\e)=\L_p-k_{ij}^2(\e)$ with asymptotic expansions (\ref{6.4}), (\ref{6.5}), (\ref{6.6}).

Let $\RE\mu_i=0$, $Q_{i,\tau}$ be not identically zero and $2r_{i,\tau}\leqslant q_i-1$. As $j=r_{i,\tau}+1,\ldots,q_i$, if
\begin{equation}\label{6.7}
\RE  (-\g_{i,\tau})^{\frac{1}{q_i-r_{i,\tau}}}e^{\frac{2\pi\iu}{q_i-r_{i,\tau}}(j-r_{i,\tau})}>0,
\end{equation}
then the pole $k_{ij}$ corresponds to an eigenvalue, while if
\begin{equation}\label{6.8}
\RE  (-\g_{i,\tau})^{\frac{1}{q_i-r_{i,\tau}}}e^{\frac{2\pi\iu}{q_i-r_{i,\tau}}(j-r_{i,\tau})}<0,
\end{equation}
the pole $k_{ij}$ corresponds to a resonance. The asymptotic expansion for this eigenvalue/resonance is given by (\ref{6.6}) if $\mu_i\ne0$ and in the case $\mu_i=0$ it reads as
\begin{equation}\label{6.9}
\l_{ij}(\e)=\L_p-\e^{2+\frac{2}{q_i-r_{i,\tau}}}(-\g_{i,\tau})^{\frac{2}{q_i-r_{i,\tau}}}
e^{\frac{4\pi\iu}{q_i-r_{i,\tau}}(j-r_{i,\tau})}
+O\big(\e^{2+\frac{3}{q_i-r_{i,\tau}}}\big).
\end{equation}
\end{theorem}

The next theorem concerns the poles emerging from {the} internal thresholds in the essential spectrum. Given $p>1$ such that $\L_p>\L_1$,  and $\tau\in\{-1,+1\}$, by $k_{ij,\tau}=k_{ij,\tau}(\e)$ we redenote the corresponding poles $k_{ij}$ of the operator $\cR_{\e,\tau}$ described in Theorems~\ref{thEmer1},~\ref{thEmer2}.

\begin{theorem}\label{thEmInt}
Let $\L_p>\L_1$ and fix $i\in\{1,\ldots,N\}$, $j\in\{1,\ldots,q_i\}$, $\tau\in\{-1,+1\}$.  Let
\begin{gather}\label{2.26a}
\RE\mu_i>0 
\\
{\text{or}}\nonumber
\\
\label{2.26b}
\begin{gathered}
\RE\mu_i=0,\qquad Q_{i,\tau}\not\equiv0,\qquad 2r_{i,\tau}\leqslant q_i-1,\\
j\in\{r_i+1,\ldots,q_i\},\qquad \RE (-\g_{i,\tau})^{\frac{1}{q_i-r_{i,\tau}}}e^{\frac{2\pi\iu}{q_i-r_{i,\tau}}(j-r_{i,\tau})}>0
\end{gathered}
\end{gather}
and
\begin{gather}\label{2.27a}
\tau\IM\mu_i<0
\\
{\text{or}}\nonumber
\\
\label{2.27b}
\begin{gathered}
\IM\mu_i=0,\qquad Q_{i,\tau}\not\equiv0,\qquad 2r_{i,\tau}\leqslant q_i-1,\\
j\in\{r_i+1,\ldots,q_i\},\qquad \tau\IM (-\g_{i,\tau})^{\frac{1}{q_i-r_{i,\tau}}}e^{\frac{2\pi\iu}{q_i-r_{i,\tau}}(j-r_{i,\tau})}<0.
\end{gathered}
\end{gather}
Then the pole $k_{ij,\tau}(\e)$ corresponds to an eigenvalue $\l_{ij,\tau}(\e)=\L_p-k_{ij,\tau}^2(\e)$ with asymptotic expansions (\ref{6.4}), (\ref{6.5}), (\ref{6.6}) {if} $\mu_i\ne0$ and asymptotic expansion (\ref{6.9}) {if}  $\mu_i=0$.

Let
\begin{gather}\label{2.28a}
\RE\mu_i<0
\\
{\text{or}}\nonumber
\\
\label{2.28b}
\begin{gathered}
\RE\mu_i=0,\qquad Q_{i,\tau}\not\equiv0,\qquad 2r_{i,\tau}\leqslant q_i-1,\\
j\in\{r_i+1,\ldots,q_i\},\qquad \RE (-\g_{i,\tau})^{\frac{1}{q_i-r_{i,\tau}}}e^{\frac{2\pi\iu}{q_i-r_{i,\tau}}(j-r_{i,\tau})}<0.
\end{gathered}
\end{gather}
Then the pole $k_{ij,\tau}(\e)$ corresponds to a resonance $\l_{ij,\tau}(\e)=\L_p-k_{ij,\tau}^2(\e)$ with asymptotic expansions (\ref{6.4}), (\ref{6.5}), (\ref{6.6}) {if} $\mu_i\ne0$ and asymptotic expansion (\ref{6.9}) {if} $\mu_i=0$.

Let $\mu_i$ be a simple eigenvalue of the matrix $\rM_1$ with an associated eigenvector $\ev_i:=(\ev_{i,1},\ldots,\ev_{i,n})$, $j=1$,
condition (\ref{2.26a}) or (\ref{2.26b}) hold, and
\begin{gather}\label{2.29a}
\tau\IM\mu_i<0
\\
{
\text{or}}\nonumber
\\
\label{2.29b}
\begin{gathered}
\IM\mu_i=0,\qquad Q_{i,\tau}\not\equiv0,\qquad r_{i,\tau}=0,
\qquad
 \tau\IM \g_{i,\tau}<0,
\end{gathered}
\end{gather}
and let there exist $s\in\{1,\ldots,p-1\}$ such that
\begin{equation}\label{2.30}
\sum\limits_{t=1}^{n}\int\limits_{\Om} e^{-K_t(0)x_d} \overline{\psi_s(x')} \cL_1 \ev_{i,t} \psi_{t-p+1}\di x\ne 0 \quad\text{or}\quad \sum\limits_{t=1}^{n}\int\limits_{\Om} e^{K_t(0)x_d} \overline{\psi_s(x')} \cL_1 \ev_{i,t} \psi_{t-p+1}\di x\ne 0.
\end{equation}
Then the pole $k_{i1,\tau}{(\e)}$ corresponds to a resonance $\l_{i,\tau}{(\e)}=\L_p-k_{i1,\tau}^2(\e)$ with the asymptotic behavior
\begin{equation*}
\l_{i,\tau}(\e)=\L_p-\e^2\mu_i^2+O(\e^4)
\end{equation*}
if $Q_{i,\tau}$ vanishes identically, and
\begin{equation*}
\l_{i,\tau}(\e)=\L_p-\e^2(\mu_i
-\e \g_{i,\tau})^2+O\big(|\mu_i|\e^4+\e^5\big)
\end{equation*}
otherwise.
\end{theorem}

\subsection{Discussion of {the} results}

In this subsection we discuss the main results formulated in Theorems~\ref{thAnCo},~\ref{thEmer1},~\ref{thEmer2},~\ref{thEmBot},~\ref{thEmInt}.
The first of them, Theorem~\ref{thAnCo}, describes 
 {a meromorphic} continuation of the resolvent of the  perturbed operator. This continuation is local and is constructed in the vicinity of the points $\L_p$, $p=1,\ldots,m$. The point $\L_1$ is the
bottom of the essential spectrum, see Lemma~\ref{lm2.1} and in vicinity of this point just one 
 {meromorphic}
continuation is possible. It is introduced as a solution to problem (\ref{2.5}) with a specified behaviour at infinity, see (\ref{2.6}), in terms of an auxiliary spectral parameter $k$. The right hand side in the equation in (\ref{2.5}) is not in the class of compactly supported functions as it is usually assumed for  {meromorphic} 
continuations, but an element of a wider space $L_2(\Om,e^{{\vt}|x_d|\di x})$. Here the presence of the weight $e^{{\vt}|x_d|}$ means that the elements of latter space in certain sense decays exponentially as $x_d\to\pm\infty$, namely they are represented as $f=e^{-\frac{{\vt}|x_d|}{2}}\tilde{f}$, where $\tilde{f}\in L_2(\Om)$. The final operator providing the 
 {meromorphic} continuation is  $\cR_{\e,\tau}$ and for $\L_1$ it is independent of $\tau$.

In the vicinity of internal thresholds $\L_p>\L_1$ in the essential spectrum, there are \emph{two different} 
{meromorphic} continuations given by the operators $\cR_{\e,-1}$ and $\cR_{\e,+1}$. The former describes 
{a meromorphic} continuation from the lower complex half-plane into the upper one, while the latter does from the upper half-plane into the lower one. In the theory of self-adjoint operators, usually only the latter continuation from the upper half-plane into the lower one is studied since it is physically meaning and it arises while considering a corresponding Cauchy problem for an evolutionary Schr\"odinger equation. However, since our perturbing operator is not assumed to be symmetric, the operator is not necessary self-adjoint. As a result, it can possess  complex eigenvalues in the vicinity of  the threshold $\L_p$. These eigenvalues are {the} poles of the resolvent of the perturbed operator. And as we shall see below, once we continue {meromorphically} 
the resolvent from the upper half-plane into the lower one, the eigenvalues in the lower half-plane can become ``invisible'' for the continuation in the sense that this continuation has no poles at such eigenvalues.  A similar situation can hold once we continue analytically the resolvent from the lower half-plane into the upper one. A clear explanation of this phenomenon is due to representations (\ref{2.6}), (\ref{2.6a}), (\ref{2.6b}) and (\ref{2.8}), (\ref{2.8a}), (\ref{2.8b}). Namely, as $\L_p>\L_1$, the functions  $\phi_{\e,j}^\pm$, $j=1,\ldots,p-1$, in (\ref{2.8a}) behave at infinity as $\phi_{\e,j}^\pm(x_d)\sim e^{-\tau K_j(k)|x_d|}$. In view of obvious identities
\begin{equation}\label{4.11}
K_j(k)=-\iu\sqrt{\L_p-\L_j}+\frac{\iu}{2\sqrt{\L_p-\L_j}}k^2+O(k^4),\qquad k\to0,
\end{equation}
the exponents $e^{-\tau K_j(k)|x_d|}$ decay only  {if} 
$\tau \IM k^2<0$. Depending on $\tau$, the latter condition means that in general only the eigenvalues either in
the upper or lower complex half-plane can serve as poles of the  {meromorphic} 
continuation of the resolvent of the operator $\Op_\e$.
This is a main reason why we deal with both  {meromorphic} 
continuations, in contrast to the case of self-adjoint operators with symmetric perturbations.

Theorems~\ref{thEmer1},~\ref{thEmer2} describe the poles of the 
 {meromorphic} continuations of the resolvent in the vicinity of the thresholds $\L_p$ in the essential spectrum. The first theorem states that in the vicinity of an $n$-multiple threshold $\L_p=\ldots=\L_{p+n-1}$ there exist exactly $n$ poles of the operator $\cR_{\e,\tau}$ counting their orders. We stress that here we count the \emph{orders} of the poles and not their \emph{multiplicities}, that is, not \emph{the number of associated linear independent solutions} to problem (\ref{2.5}) with $f=0$, $k=k_{ij}(\e)$.  The multiplicity of each pole does not exceed its order; this can be shown by the technique used in the proofs of Lemmata~6.2,~6.3 in \cite{MPAG07} and Lemmata~6.2,~6.3 in \cite{AHP07}. However, in general, the multiplicities and the orders coincide only if the perturbation $\cL(\e)$ is symmetric. The reason is that in the general case of a non-symmetric perturbation, in a certain matrix controlling the structure of the poles $k_{ij}$, a non-diagonal Jordan block can arise and this gives rise to adjoint vectors instead of the eigenvectors, see Section~\ref{ss:Emer} and the calculations involving matrix $\rM_{\e,\tau}$. Of course, the multiplicity of each pole $k_{ij,\tau}$ is at least one. In particular, if all poles $k_{ij,\tau}$ are different for a fixed $\tau$, the total multiplicity is equal to $n$. Theorem~\ref{thEmer1} provides leading  terms in the asymptotic expansions for the poles $k_{ij,\tau}$, while Theorem~\ref{thEmer2} specifies these expansions. In some cases it just improves the estimate for the error terms, see (\ref{2.13}), (\ref{2.17}), (\ref{2.18}), while in some cases, a next-to-leading term in the expansions can be found, see (\ref{2.19}). Theorems~\ref{thEmer1},~\ref{thEmer2} treat a general case, when the eigenvalues of the matrix $\rM_1$ are of arbitrary multiplicities and no extra assumptions are made for the matrix $\rM_{2,\tau}$. In an important particular case, when $\mu_i$ is a simple eigenvalue of the matrix $\rM_1$, we have $q_i=1$ and $r_{i,\tau}=0$. In this case there exists just one pole $k_{i1,\tau}$ with asymptotic behavior (\ref{2.10}) and expansion (\ref{2.13}), (\ref{2.19}) can be applied, which yields that
\begin{equation*}
k_{i1,\tau}(\e)=\e\mu_i -\e^2 \g_{i,\tau}
+O(\e^3).
\end{equation*}
If $n=1$, that is, $\L_p$ is a simple eigenvalue of the operator $\Op'$, the above expansions can be specified as follows:
\begin{equation}\label{1.2}
k_{i1,\tau}(\e)=-\frac{\e}{2}\int\limits_{\Om} \overline{\psi_p} \cL_1\psi_{p}\di x -\frac{\e^2}{2}\int\limits_{\Om} \overline{\psi_p} (\cL_2-\cL_1 \cG_{p,\tau} \cL_1)\psi_p\di x
+O(\e^3).
\end{equation}
We also observe that since the operator $\cR_{\e,\tau}$ is independent of $\tau$ {if} 
$p=1$, {in the general situation} there are only $n$ poles
in the vicinity of the bottom $\L_1$ of the essential spectrum. In the vicinity of internal thresholds $\L_p>\L_1$, the operators $\cR_{\e,\tau}$ depend on $\tau$ and this is why there are $2n$ poles in the vicinity of $\L_p$. In particular, if $\L_p$ is an $n$-multiple eigenvalue of the operator $\Op'$, there can be $2n$ different simple eigenvalues of the operator $\Op_\e$ converging to $\L_p$, see examples in Subsection~\ref{ss:Examples}.

The above discussed poles of the operators $\cR_{\e,\tau}$ correspond either to the eigenvalues or resonances depending on the behavior of the associated non-trivial solutions. This behaviour is completely described by formulae (\ref{2.8}), (\ref{2.8a}) and we just need  to identify whether the function $\psi_\e$ decays exponentially at infinity or not. In the former case we deal with an eigenvalue, otherwise with a resonance. As we see, the functions $\phi_{\e,j}^\pm(x_d)$ decay  exponentially as $j\geqslant p+n$ no matter how the corresponding pole looks like. 
{But  
for $j=p,\ldots, p+n-1$ these functions} behave at infinity as $\phi_{\e,j}^\pm(x_d,k)\sim e^{-k_\e|x_d|}$. {These functions decay} exponentially as $\RE k_\e>0$, is periodic as $\RE k_\e=0$ and grows exponentially as $\RE k_\e<0$. If $\L_p>\L_1$, we also have to control the behavior of the functions $\phi_{\e,j}^\pm(x_d)$ with $j=1,\ldots,p-1$. This is easily done by identities (\ref{4.11}): the functions $\phi_{\e,j}^\pm(x_d)$, $j=1,\ldots,p-1$,  decay  exponentially 
 {if} $\tau\IM k_\e^2<0$, are periodic {if}  
 $\IM k_\e^2=0$ and grow exponentially 
 {if}  $\tau\IM k_\e^2>0$. All discussed conditions can be checked by means of asymptotic expansions provided by Theorems~\ref{thEmer1},~\ref{thEmer2} for a given pole. And exactly this is done in the proof of Theorems~\ref{thEmBot},~\ref{thEmInt}. Conditions in Theorem~\ref{thEmBot} are aimed {at} 
checking the sign of the real part of a given pole {and proving at the same time that} 
at least one of the coefficients $c_{\e,j}^\pm$, $j=1,\ldots,n$, in (\ref{2.8a}) is non-zero. Similar conditions (\ref{2.26a}), (\ref{2.26b}), (\ref{2.27a}), (\ref{2.27b}), (\ref{2.28a}), (\ref{2.28b}) ensure that the functions $\phi_{\e,j}^\pm$, $j=1,\ldots,p-1$  decay exponentially, that is, $\tau\IM k_\e^2<0$, while for $j=p,\ldots,p+n-1$, these functions demonstrate either an exponential decay  or an exponential growth. Conditions~(\ref{2.29a}),~(\ref{2.29b}),~(\ref{2.30}) describe a more gentle situation. Namely, here the real part of the pole is negative and the functions $\phi_{\e,j}^\pm$, $j=p,\ldots,p+n-1$, decay exponentially. However, $\tau\IM k_\e^2>0$ and this means that the functions $\phi_{\e,j}^\pm$, $j=1,\ldots,p-1$, can grow exponentially. This is true, once we guarantee that at least one of the coefficients $c_{\e,j}^\pm$, $j=1,\ldots,p-1$, is non-zero. This is indeed the case thanks to condition~(\ref{2.30}). The asymptotic expansions for the eigenvalues and the resonances provided in Theorems~\ref{thEmBot},~\ref{thEmInt} are implied immediately by the formula $\l_\e=\L_p-k_\e^2$ relating the eigenvalues/resonances with a pole $k_\e$ and the asymptotic expansions for the poles stated in Theorems~\ref{thEmer1},~\ref{thEmer2}.

{Let us briefly discuss the main ideas underlying our main results. First, we rather straightforwardly construct  the meromorphic continuation for the resolvent of the unperturbed operator. Namely, we find explicitly the projection of the solution to problem (\ref{2.5}) on the eigenfunctions $\psi_j$, $j=1,\ldots,m$, and study  then the properties of the coefficients in this projection and of the remaining orthogonal part in the solution. Once such continuation is constructed, for proving our main results, we apply an approach being a modification of the technique suggested in \cite{MSb06}, \cite{Izv08}, \cite{Ga1}, \cite{Ga2}. The idea is to regard the perturbation as a right-hand side and to apply then the meromorphic continuation of the unperturbed operator. After some simple calculations this leads us to an operator equation 
with a certain finite rank perturbation. Resolving this equation, we rather easily succeed to construct the meromorphic continuation for the perturbed operator and identify its poles as solutions to a nonlinear eigenvalue problem for  some explicitly calculated matrix depending also on the small parameter. Analysing then this problem by means of methods from the theory of complex functions, we 
study the existence of the poles and their asymptotic behavior.}

{Although we restrict ourselves by considering second order differential operators, our approach can be also adapted for certain operators of higher order. However, general higher order operators can have a richer spectral structure of the edges in the essential spectrum and there can be more complicated scenarios of their bifurcations under perturbations, see, for instance, \cite{Na2}. This is why the case  of second order operators deserves a separate study, what is done in the present work.} Our results are of general nature and are applicable to  wide classes of unperturbed operators and perturbations. In the next section we discuss some possible examples of both unperturbed operators and perturbations as well as some specific examples motivated by physical models.

\section{Examples}
\label{sec:examples}

In this section we provide examples demonstrating our
main results.

\subsection{Unperturbed operator}
 Here we discuss some examples of the unperturbed operator, namely, of the operator $\Op$. This is a general self-adjoint second order differential operator and it includes such classical operators as a Schr\"odinger operator:
\begin{equation*}
\Op=-\D+A_0,\qquad A_0= A_0(x'),
\end{equation*}
a magnetic Schr\"odinger operator:
\begin{equation*}
\Op=(\iu \nabla_{x'}+{A})^2-\frac{\p^2\ }{\p x_d^2}+A_0,\qquad {A=(A_1,\ldots,A_{d-1})},\qquad {A_j=A_j(x')},\qquad A_0= A_0(x'),
\end{equation*}
a Schr\"odinger operator with metric:
\begin{equation*}
\Op=-\sum\limits_{i,j=1}^{d-1}\frac{\p\ }{\p x_i} A_{ij} \frac{\p\ }{\p x_j}-\frac{\p^2\ }{\p x_d^2}+A_0, \qquad A_{ij}= A_{ij}(x'),\qquad A_0= A_0(x').
\end{equation*}
All these operators are considered in a tubular domain $\Om=\om\times\mathds{R}$. If $\om=\mathds{R}^{d-1}$, then the domain $\Om$ becomes an entire space $\mathds{R}^d$. If $\om$ is a bounded domain, not necessary connected, then $\Om$ is an infinite cylinder, which is to be regarded as a quantum waveguide if the Dirichlet condition is imposed on its boundary and as an acoustic waveguide if the boundary is subject to the Neumann condition. Further examples of unbounded domains $\om$
are also possible. For instance, if $\om$ is the half-space $\om:=\{x':\, x_j>0\}$ in $\mathds{R}^{d-1}$ for some $j=1,\ldots,d-1$, the domain $\Om$ becomes the half-space $\{x:\, x_j>0\}$ in $\mathds{R}^d$. We can also consider a more complicated domain $\om:=\{x': x_{d-1}<h(x_1,\ldots,x_{d-2})\}$ for some smooth function $h$, then $\Om=\{x:\, x_{d-1}<h(x_1,\ldots,x_{d-2}), x_d\in\mathds{R}\}$.

\subsection{Perturbation}

In this subsection we discuss possible examples of the perturbing operator $\cL(\e)$. The first example is a second order differential operator:
\begin{equation}\label{1.3}
\cL(\e)=\sum\limits_{i,j=1}^{n} {\Ups}_{ij}(x,\e)\frac{\p^2\ }{\p x_i\p x_j}+ \sum\limits_{j=1}^{n} {\Ups}_j(x,\e)\frac{\p\ }{\p x_j} + {\Ups}_0(x,\e).
\end{equation}
Here ${\Ups}_{ij},\, {\Ups}_j,\, {\Ups}_0\in L_\infty(\Om)$ are some functions, not necessarily real-valued, satisfying the representations
\begin{equation}\label{1.3a}
{\Ups}_\natural(x,\e)={\Ups}_\natural^{(1)}(x)+\e {\Ups}_\natural^{(2)}(x)+\e^2 {\Ups}_\natural^{(3)}(x,\e),\qquad \natural=ij,j,0,
\end{equation}
where ${\Ups}_\natural^{(s)}\in L_\infty(\Om)$ are some functions obeying
the estimates:
\begin{equation*}
\|{\Ups}_\natural^{(s)}e^{2{\vt}|x_d|}\|_{L_\infty(\Om)}<C,
\end{equation*}
and {$\vt$} and $C$ are some fixed positive constant independent of $\e$. In this case the operators $\cL_s$ read as
\begin{equation*}
\cL_s:=\sum\limits_{i,j=1}^{n} {\Ups}_{ij}^{(s)}\frac{\p^2\ }{\p x_i\p x_j}+ \sum\limits_{j=1}^{n} {\Ups}_j^{(s)}\frac{\p\ }{\p x_j} + {\Ups}_0^{(s)},\qquad s=1,2,3.
\end{equation*}
Particular cases of this example  are small potential, small magnetic field, small metric.

The second example is an integral operator of the form
\begin{equation*}
(\cL(\e)u)(x,\e)=\int\limits_{\Om} {J}(x,y,\e)u(y)\di y,
\end{equation*}
where ${J}\in L_2(\Om\times\Om)$ is some kernel, not necessarily real-valued and symmetric and satisfying the representation:
\begin{equation*}
{J}(x,y,\e)={J}_1(x,y)+\e {J}_2(x,y)+\e^2 {J}_3(x,y,\e),
\end{equation*}
where $L_i$ are some functions obeying the estimates:
\begin{equation*}
\int\limits_{\Om\times\Om}  |{J}_i|^2 e^{{vt}(|x_d|+|y_d|)}\di x\di y<C,\qquad i=1,2,3,
\end{equation*}
and ${vt}$ and $C$ are some fixed positive constant independent of $\e$.

The third example is a localized $\d$-interaction with a complex-valued density. Namely, let $S\subset\Om$ be {a} 
manifold of codimension $1$ and of smoothness $C^3$. We assume that it is compact and has no edge. The perturbed operator in question is
\begin{equation*}
\Op_\e=-\D+\e\b\d(x-S),
\end{equation*}
which acts as $\Op_\e u=-\D u$ on the domain $\Dom(\Op_\e)$ formed by the functions $u\in W_2^2(\Om\setminus S)\cap W_2^1(\Om)$ obeying the boundary condition $\cB u=0$ on $\p\Om$ and the boundary conditions
\begin{equation*}
[u]_S=0,\qquad \left[\frac{\p u}{\p\nu}\right]_S=\e\b u.
\end{equation*}
Here $[u]_S$ denotes the jump of the function on $S$, namely,
\begin{equation*}
[u]_S:={\lim\limits_{t\to0+}\big(u(\cdot+t\nu)-u(\cdot-t\nu)\big),}
\end{equation*}
and $\nu$ is the unit normal to $S$ directed outside the domain enveloped by $S$. By $\b$ we denote some complex-valued function defined on $S$, uniformly bounded and belonging to $C^2(S)$ . Such operator does not satisfy our assumptions for $\cL(\e)$ since now the perturbation changes the domain. However, it is possible to reduce the perturbed operator to another one obeying needed assumptions and having the same eigenvalues and resonances. Namely, thanks to the made assumptions on the manifold $S$, in a small vicinity of the $S$ we can introduce a new variable $\rho$ being the distance from a point to $S$ measured along the normal $\nu$. This variable is well-defined at least in the neighbourhood $\{x:\, \dist(x,S)<\rho_0\}$  of $S$, where $\rho_0$ is some fixed number. Let $\chi=\chi(\rho)$ be an infinitely differentiable function such that $\chi(\rho)=\frac{|\rho|}{2}$ as $|\rho|<\frac{\rho_0}{3}$ and $\chi(\rho)=0$ as $|\rho|>\frac{2\rho_0}{3}$. By $\cU_\e$ we denote the multiplication operator $\cU_\e u:=(1+\e\b\chi)^{-1} u$. It is straightforward to check that this operator maps the domain of the operator $\Op_\e$ onto the space $\{u\in W_2^2(\Om):\, \cB u=0\ \text{on}\ \p\Om\}$. This space serves as the domain for an operator $\tilde{\Op}_\e:=\cU_\e\Op_\e\cU_\e^{-1}$. It is straightforward to confirm that the differential operator for the latter operator reads as
\begin{equation*}
\tilde{\Op}_\e=-\D+\e\cL(\e),\qquad \cL(\e):=-2(1+\e\b\chi)^{-1}\nabla\b\chi\cdot\nabla-(1+\e\b\chi)^{-1}\D\b\chi
\end{equation*}
and we see that a first order differential operator $\cL(\e)$ is a particular case of operator (\ref{1.3}). It is also clear that the operators $\Op_\e$ and $\tilde{\Op}_\e$ have the same eigenvalues and resonances since the operator $\cU_\e$ does not change the behavior of the functions at infinity. Hence, we can study the eigenvalues and resonances of the operator $\tilde{\Op}_\e$ and transfer then the results to the operator $\Op_\e$.

Our fourth example is a geometric perturbation. Namely, let $\om$ have a non-empty boundary, then the same is true for $\Om$. By $\G$ we denote a bounded subset of the boundary $\p\Om$. Let $\rho$ be a distance to a point measured along the outward normal to $\p\Om$ and $h\in C^2(\p\Om)$ be some real function defined on $\G$ and compactly supported in $\G$. Then we consider a domain $\Om_\e$ obtained by a small variation of the part $\G$ of the boundary $\p\Om$. Namely, $\Om_\e$ is a domain with the following boundary
\begin{equation*}
\p\Om_\e:=(\p\Om\setminus\G)\cup\{x:\, \rho=\e h\}.
\end{equation*}
In such domain we consider an operator   with differential expression (\ref{2.0}) subject to the Dirichlet boundary condition or Neumann condition. We assume that all the coefficients in the differential expression depend on $x'$ only and are infinitely differentiable.
Such perturbed operator does not fit our scheme since here the domain $\Om_\e$ depends on $\e$. However, as in the previous example, it is possible to transform such operator to another one fitting our assumptions. Namely,
let $\chi=\chi(x)$ be an infinitely differentiable cut-off function equalling to one in some fixed sufficiently small $d$-dimensional neighbourhood of $\G$ and vanishing outside some bigger neighbourhood. In this bigger neighbourhood we introduce local coordinates $(P,\rho)$, where $P\in\p\Om$. A point $x$ is recovered from $(P,\rho)$ by measuring the distance $\rho$ along the outward normal to $\p\Om$ at the point $P$.   Then we define a mapping $\cP$ by the following rule: for each point $x$, we find corresponding $(P,\rho)$ and the action of the mapping is a point corresponding to $(P,\rho-\e h(P))$. We introduce new coordinates by the formula $\tilde{x}:=x(1-\chi(x))+\e\chi(x)\cP(x)$. It is easy to see that these coordinates are well-defined provided $\e$ is small enough and after passing to these new coordinates, the domain $\Om_\e$ transforms into $\Om$, while the operator $\Op_\e$ becomes $\Op+\cL(\e)$, where $\cL(\e)$ is some second order differential operator of form (\ref{1.3}) with compactly supported coefficients
obeying (\ref{1.3a}).

\subsection{Emerging poles for particular models}\label{ss:Examples}

In this section we apply Theorems~\ref{thEmBot},~\ref{thEmInt} to some simple two- and three-dimensional operators motivated by {an} interesting physical background.

\subsubsection{Planar waveguide}

The first model is an infinite planar waveguide modeled by the Dirichlet Laplacian. Namely, we let $d=2$, $\om:=(0,\pi)$, and $\Op'=-\frac{d^2\ }{dx_2^2}$ subject to the Dirichlet boundary condition. Then $\Om:=\{x:\, 0<x_1<\pi\}$ is an infinite strip and $\Op=-\D$ is the Dirichlet Laplacian in $\Om$. As a perturbation, we choose a complex-valued potential of {the} form $\cL(\e):=V_1+\e V_2$, where $V_i=V_i(x)$ are some continuous compactly supported complex-valued functions.
The operator $\Op'$ has a purely discrete spectrum formed by simple eigenvalues $\L_p:=p^2$, $p\in\mathds{N}$, and the associated eigenfunctions normalized in $L_2(0,\pi)$ are $\psi_p(x_d):=\frac{\sqrt{2}}{\sqrt{\pi}}\sin p x_1$.

This situation models a slab optical waveguide of  a  finite width, where the cladding in direction $x_1$ imposes zero boundary conditions at $x_1=0$ and $x_1=\pi$. Assuming that the waveguide is infinite in the second direction, a paraxial diffraction of an incident beam can be described using the normalized  equation   in the form (see e.g. \cite{Kivshar})
\begin{equation}
\label{optics}
\iu\partial_z \Phi + \D
\Phi - \e V_\e(x)\Phi = 0,
\end{equation}
where $\Phi(x_1,x_2, z)$ corresponds to complex amplitude of  the electrical field,    the optical potential $V_\e(x)$ describes a weak localized modulation of the complex-valued refractive index, and $z$ is the direction of propagation of the pulse. For stationary modes $\Phi = e^{-\iu\l z} \psi$, where $-\lambda$ has the meaning of    propagation constant, equation (\ref{optics})  reduces to the eigenvalue problem in the above described planar waveguide for the equation
\begin{equation*}
-\D\psi+\e V_\e\psi=\l\psi.
\end{equation*}
This is exactly the mathematical model we formulated above once we let $V_\e=V_1+\e V_2$.
Let us consider the bifurcation of the thresholds $p^2$ under the presence of  a small localized potential $V_\e$.

For $p=1$, there is just one pole $k_\e$ and according formula (\ref{1.2}), its asymptotic behavior reads as
\begin{equation}\label{1.4}
k_\e=-\frac{\e}{\pi}\int\limits_{\Om} V_1(x)\sin^2 x_1\di x -\frac{\e^2}{\pi} \int\limits_{\Om} \big(V_2(x)\sin^2 x_1 -V_1(x) U(x)\sin x_1 \big)\di x
+O(\e^3),
\end{equation}
where $U:=\cG_{1,\tau} (V_1\sin x_1)$. This function is given by formula (\ref{3.9}). The term $U^\bot:=((\Op^\bot-\L_1)^{-1}f^\bot)(x)$  with $f=V_1\sin x_1$ and $f^\bot$ defined by (\ref{3.11})
solves the boundary value problem
\begin{equation*}
(-\D-1)U^\bot=V_1 \sin x_1-\frac{2}{\pi}\sin x_1\int\limits_0^\pi V_1(t_1,x_2)\sin^2 t_1\di t_1 \quad\text{in}\quad\Om,\qquad U^\bot=0\quad\text{on}\quad\p\Om.
\end{equation*}
This problem can be solved explicitly by the separation of variables
and this gives the final formula for $U$:
\begin{align*}
U(x)=&-\frac{1}{\pi}\sin x_1 \int\limits_{\Om} |x_2-t_2|V_1(t_1,t_2)\sin^2 t_1\di t
\\
&+
\sum\limits_{j=2}^{\infty}
\frac{\sin j x_1}{\pi\sqrt{j^2-1}}
\int\limits_{\Om} e^{-\sqrt{j^2-1}|x_2-t_2|}V_1(t_1,t_2)\sin t_1\sin j t_1\di t.
\end{align*}
Hence,
\begin{align*}
\int\limits_{\Om}\big(V_2(x)&\sin^2 x_1 -V_1(x) U(x)\sin x_1 \big)\di x
\\
=& \int\limits_{\Om}  V_2(x)\sin^2 x_1 \di x
 +\frac{1}{\pi}\int\limits_{\Om^2} |x_2-t_2|V_1(x)V_1(t)\sin^2 t_1\sin^2 x_1\di t\di x
\\
&- \sum\limits_{j=2}^{\infty}
\frac{1}{\pi\sqrt{j^2-1}}
\int\limits_{\Om^2} e^{-\sqrt{j^2-1}|x_2-t_2|}V_1(t)V_1(x)\sin t_1\sin j t_1\sin  x_1\sin j x_1\di t\di x.
\end{align*}
Now we apply Theorem~\ref{thEmBot} and we see that if
\begin{equation*}
\RE\int\limits_{\Om} V_1(x)\sin^2 x_1\di x<0
\end{equation*}
or
\begin{equation*}
\RE\int\limits_{\Om} V_1(x)\sin^2 x_1\di x=0 \quad\text{and}\quad \RE \int\limits_{\Om} \big(V_2(x)\sin^2 x_1 -V_1(x) U(x)\sin x_1\big)\di x<0,
\end{equation*}
then the pole $k_\e$ corresponds to an eigenvalue $\l_\e=1-k_\e^2$.
And if
\begin{equation*}
\RE\int\limits_{\Om} V_1(x)\sin^2 x_1\di x>0
\end{equation*}
or
\begin{equation*}
\RE\int\limits_{\Om} V_1(x)\sin^2 x_1\di x=0 \quad\text{and}\quad \RE \int\limits_{\Om} \big(V_2(x)\sin^2 x_1 -V_1(x) U(x)\sin x_1 \big)\di x>0,
\end{equation*}
then the pole $k_\e$ corresponds to a resonance $\l_\e=1-k_\e^2$. The asymptotic expansion for this eigenvalue/resonance is given by (\ref{6.6}), (\ref{6.9}) but it is more straightforward to find it by (\ref{1.4}) and the above formula for $\l_\e$.

We proceed to the case $p>1$. Here we again apply formula (\ref{1.2}) to obtain
\begin{equation}\label{1.9}
k_{\e,\tau}=-\frac{\e}{\pi}\int\limits_{\Om} V_1(x)\sin^2 p x_1\di x -\frac{\e^2}{\pi} \int\limits_{\Om} \big(V_2(x)\sin^2 p x_1 -V_1(x) U_\tau(x)\sin p x_1\big)\di x
+O(\e^3),
\end{equation}
where $U_\tau:=\cG_\tau V_1\sin p x_1$ is given by formula (\ref{3.9}). The term $U_\tau^\bot:=((\Op^\bot-\L_p)^{-1}f^\bot)(x)$  with $f=V_1\sin p x_1$ and $f^\bot$ defined by (\ref{3.11})
solves the boundary value problem
\begin{equation*}
(-\D-p^2)U_\tau^\bot=V_1 \sin p x_1-\frac{2}{\pi}\sin p x_1\int\limits_0^\pi V_1(t_1,x_2)\sin^2 p t_1\di t_1 \quad\text{in}\quad\Om,\qquad U_\tau^\bot=0\quad\text{on}\quad\p\Om.
\end{equation*}
The solution is again given by the separation of variables and a final formula for $U_\tau$ reads as
\begin{equation}\label{1.8}
\begin{aligned}
&U_\tau(x)=\sum\limits_{j=1}^{p-1}
\frac{\iu\tau\sin j x_1}{\pi\sqrt{p^2- j^2}}
\int\limits_{\mathds{R}} e^{\iu\tau\sqrt{p^2- j^2}|x_2-t_2|}U_j(t_2)\di t_2
\\
&\hphantom{U_\tau(x)=\sum\limits_{j=1}^{p-1}} -\frac{1}{\pi}\sin p x_1 \int\limits_{\mathds{R}} |x_2-t_2|U_p(t_2)\di t_2
\\
&\hphantom{U_\tau(x)=\sum\limits_{j=1}^{p-1}}+
\sum\limits_{j=p+1}^{\infty}
\frac{\sin j x_1}{\pi\sqrt{j^2-p^2}}
\int\limits_{\mathds{R}} e^{-\sqrt{j^2-p^2}|x_2-t_2|}U_j(t_2)\di t_2.
\\
&U_j(x_2):=
\int\limits_0^\pi V_1(t_1,x_2)\sin p t_1\sin j t_1\di t,\qquad j\ne p,
\end{aligned}
\end{equation}
Then we get:
\begin{equation}\label{1.19}
\begin{aligned}
 \int\limits_{\Om} V_1(x) U_\tau(x)\sin p x_1 \di x=
 &\sum\limits_{j=1}^{p-1}
\int\limits_{\mathds{R}^2} \frac{\iu\tau e^{\iu\tau\sqrt{p^2- j^2}|x_2-t_2|}}{\pi\sqrt{p^2- j^2}}U_j(x_2)U_j(t_2)\di x_2\di t_2
\\
&\hphantom{\sum\limits_{j=1}^{p-1}} -\frac{1}{\pi}  \int\limits_{\mathds{R}^2} |x_2-t_2|U_p(x_2)U_p(t_2)\di x_2\di t_2
\\
&+
\sum\limits_{j=p+1}^{\infty}
\int\limits_{\mathds{R}^2} \frac{e^{-\sqrt{j^2-p^2}|x_2-t_2|}}{\pi\sqrt{j^2-p^2}}U_j(x_2)U_j(t_2)\di t_2\di x_2.
\end{aligned}
\end{equation}
Now we can apply Theorem~\ref{thEmInt} for $\tau=+1$ and $\tau=-1$ and to determine whether the  poles  $k_{\e,\tau}$ correspond to eigenvalues or resonances. As we see, in a general situation we can have two eigenvalues or two resonances or one eigenvalue and one resonance. Let us show that each of these situations is possible.

First of all we observe that in notations of Theorem~\ref{thEmInt} we have $N=1$, $q_i=1$, $r_i=0$,
\begin{gather*}
\mu_1=-\frac{1}{\pi}\int\limits_{\Om} V_1(x)\sin^2 p x_1\di x,
\\
(-\g_{i,\tau})^{\frac{1}{q_i-r_{i,\tau}}}e^{\frac{\pi\iu}{q_i-r_{i,\tau}}(j-r_{i,\tau})}=
-\frac{1}{\pi} \int\limits_{\Om} \big(V_2(x)\sin^2 p x_1 -V_1(x) U_{1,\tau}(x)\sin p x_1\big)\di x.
\end{gather*}
Assume now that $V_1$ is a complex-valued potential such that
\begin{equation*}
\RE \int\limits_{\Om} V_1(x)\sin^2 p x_1\di x>0.
\end{equation*}
Then condition (\ref{2.28a}) is satisfied and both poles $k_{\e,\tau}$, $\tau=\{-1,+1\}$, correspond to resonances.

If
\begin{equation*}
\RE \int\limits_{\Om} V_1(x)\sin^2 p x_1\di x<0,\qquad
\IM \int\limits_{\Om} V_1(x)\sin^2 p x_1\di x\ne0,
\end{equation*}
then conditions (\ref{2.26a}), (\ref{2.27a}) are satisfied with \begin{equation*}
\tau=\sgn \IM \int\limits_{\Om} V_1(x)\sin^2 p x_1\di x
\end{equation*}
 and for such $\tau$, the pole $k_{\e,\tau}$ corresponds to an eigenvalue. If in addition,
\begin{equation*}
\int\limits_{\Om} e^{-\iu\sqrt{p^2-s^2} x_2} V_1(x)\sin s x_1 \sin px_1\di x\ne0 \quad\text{or}\quad\int\limits_{\Om} e^{\iu\sqrt{p^2-s^2} x_2} V_1(x)\sin sx_1 \sin px_1\di x\ne0,
\end{equation*}
for some $s\in\{1,\ldots,p-1\}$, then conditions (\ref{2.29a}), (\ref{2.30}) are satisfied and the pole $k_{\e,\tau}$ with
\begin{equation*}
\tau=-\sgn \IM \int\limits_{\Om} V_1(x)\sin^2 p x_1\di x
\end{equation*}
corresponds to a resonance.

In order to realize a situation with two eigenvalues, we assume that $V_2=0$ and consider a special class of $\mathcal{PT}$-symmetric  potentials $V_1$. Namely, we suppose that
\begin{equation}\label{1.23}
V_1(x)=W_1(x)+\iu W_2(x),
\end{equation}
where $W_1$, $W_2$ are real-valued compactly supported potentials with certain parity:
\begin{equation}\label{1.20}
W_1(x_1,-x_2)=W_1(x_1,x_2),\qquad
W_2(x_1,-x_2)=-W_2(x_1,x_2).
\end{equation}
These assumptions yield that the operator $\Op_\e$ is $\mathcal{PT}$-symmetric (or partially $\PT$-symmetric using the terminology from \cite{PPT}). They also imply immediately that
\begin{equation*}
\int\limits_{\Om} V_1(x)\sin^2 p x_1\di x=
\int\limits_{\Om} W_1(x)\sin^2 p x_1\di x
\end{equation*}
and we assume that
\begin{equation}\label{1.25}
\int\limits_{\Om} W_1(x)\sin^2 p x_1\di x<0.
\end{equation}

It follows from assumptions (\ref{1.20}) and  the definition of the functions $U_j$  in (\ref{1.8}) that these functions
are given by the formulae
\begin{equation*}
U_j(x_2):=W_{1,j}(x_2)+\iu W_{2,j}(x_2),\qquad
W_{s,j}(x_2):=\int\limits_0^\pi W_s(t_1,x_2)\sin p t_1\sin j t_1\di t,\qquad s=1,2,
\end{equation*}
and the functions $W_{1,j}$ are even, while $W_{2,j}$ are odd.
As $j\geqslant p+1$, by making the change of the variables $x_2\mapsto -x_2$, $t_2\mapsto -t_2$ in the integrals in the  second sum in (\ref{1.19}), we get:
\begin{equation*}
\int\limits_{\mathds{R}^2} \frac{e^{-\sqrt{j^2-p^2}|x_2-t_2|}}{\pi\sqrt{j^2-p^2}}U_j(x_2)U_j(t_2)\di t_2\di x_2=\int\limits_{\mathds{R}^2} \frac{e^{-\sqrt{j^2-p^2}|x_2-t_2|}}{\pi\sqrt{j^2-p^2}}\overline{U_j(x_2)}\overline{U_j(t_2)}\di t_2\di x_2
\end{equation*}
and hence,
\begin{equation*}
\IM \sum\limits_{j=p+1}^{\infty}
\int\limits_{\mathds{R}^2} \frac{e^{-\sqrt{j^2-p^2}|x_2-t_2|}}{\pi\sqrt{j^2-p^2}}U_j(x_2)U_j(t_2)\di t_2\di x_2=0.
\end{equation*}
In the same way we confirm that
\begin{align*}
&\IM \sum\limits_{j=1}^{p-1}
\int\limits_{\mathds{R}^2} \frac{\tau\sin \tau\sqrt{p^2- j^2}|x_2-t_2|}{\pi\sqrt{p^2- j^2}}U_j(x_2)U_j(t_2)\di t_2\di x_2=0,
\\
&\IM \frac{1}{\pi}  \int\limits_{\mathds{R}^2} |x_2-t_2|U_p(x_2)U_p(t_2)\di x_2\di t_2=0.
\end{align*}
Hence, by two above identities and (\ref{1.19}),
\begin{align*}
\IM
\int\limits_{\Om} &V_1(x) U_\tau(x)\sin p x_1 \di x=\IM \sum\limits_{j=1}^{p-1}
\int\limits_{\mathds{R}^2} \frac{\iu\tau\cos\tau\sqrt{p^2- j^2}|x_2-t_2|}{\pi\tau\sqrt{p^2-
j^2}}U_j(x_2)U_j(t_2)\di t_2\di x_2
\\
=&\tau\RE \sum\limits_{j=1}^{p-1}
\int\limits_{\mathds{R}^2} \frac{\cos\sqrt{p^2- j^2}|x_2-t_2|}{\pi\sqrt{p^2-
j^2}}U_j(x_2)U_j(t_2)\di t_2\di x_2
\\
=&\tau \sum\limits_{j=1}^{p-1}
\int\limits_{\mathds{R}^2} \frac{\cos\sqrt{p^2- j^2}(x_2-t_2)
}{\pi\sqrt{p^2-
j^2}}\big(W_{1,j}(x_2)W_{1,j}(t_2)-W_{2,j}(x_2)W_{2,j}(t_2)\big)\di t_2\di x_2.
\end{align*}
By straightforward calculations, for an arbitrary compactly supported function $W(x_2)$  we obtain:
\begin{align*}
\int\limits_{\mathds{R}^2}  \cos\sqrt{p^2- j^2}(x_2-t_2)
 W (x_2)W (t_2) \di t_2\di x_2=&\left(\int\limits_{\mathds{R}}  W (x_2) \cos\sqrt{p^2- j^2} x_2
\di x_2\right)^2
\\
&+ \left(\int\limits_{\mathds{R}}
 W (x_2) \sin\sqrt{p^2- j^2} x_2 \di x_2\right)^2.
\end{align*}
Hence, by two latter identities  and (\ref{1.20}),
\begin{align*}
\tau\IM
\int\limits_{\Om} V_1(x) U_\tau(x)\sin p x_1 \di x=\frac{1}{\pi} \sum\limits_{j=1}^{p-1} \frac{1}{\sqrt{p^2-
j^2}}&\left(\left(\int\limits_{\mathds{R}}
 W_{1,j} (x_2) \cos\sqrt{p^2- j^2} x_2 \di x_2\right)^2
\right.
\\
&\left.\hphantom{\Bigg(} - \left(\int\limits_{\mathds{R}}  W_{2,j} (x_2) \sin\sqrt{p^2- j^2} x_2
\di x_2\right)^2\right).
\end{align*}
The latter formula implies that the sign of its left hand side is the same for both $\tau\in\{-1,+1\}$ and we can make this sign being $-1$ by choosing appropriately $W_2$ once we fix $W_1$ satisfying (\ref{1.25}), namely, we can satisfy the condition
\begin{equation}
\begin{aligned}  \sum\limits_{j=1}^{p-1} \frac{1}{\sqrt{p^2-
j^2}}&\left(\left(\int\limits_{\mathds{R}}
 W_{1,j} (x_2) \cos\sqrt{p^2- j^2} x_2 \di x_2\right)^2
\right.
\\
&\left. \hphantom{\Bigg(} - \left(\int\limits_{\mathds{R}}  W_{2,j} (x_2) \sin\sqrt{p^2- j^2} x_2
\di x_2\right)^2\right)<0
\end{aligned}\label{1.28}
\end{equation}
For instance, this can be done by letting $W_2=\a \tilde{W}_2$ with a sufficient large $\a$, where $\tilde{W}_2=\tilde{W}_2(x)$ is a real odd compactly supported function such that
\begin{equation*}
\sum\limits_{j=1}^{p-1} \frac{1}{\sqrt{p^2-
j^2}}\left(\int\limits_{\Om} \tilde{W}_2 (x)\sin p x_1\sin j x_1 \sin\sqrt{p^2- j^2} x_2
\di x\right)^2>0.
\end{equation*}
Once conditions~(\ref{1.25}),~(\ref{1.28}) hold,  Theorem~\ref{thEmInt} states that both poles $k_{\e,\tau}$, $\tau\in\{-1,+1\}$ correspond to the eigenvalues located in the vicinity of the internal threshold $\L_p$.

The above analytic results, namely, the discusse{d}  asymptotic expansions, approximate well the true eigenvalues and resonances {for sufficiently small} 
$\e$. 
In order to demonstrate how small $\e$ is to be chosen, we make some numerical computations.

For numerics we use
\begin{equation*}
 W_1(x) = -\sum_{j=1}^3 a_j\sin jx_1\cos\frac{x_2}{2}, \qquad W_2(x) =\sum_{j=1}^3 b_j\sin jx_1\sin x_2,
\end{equation*}
where $a_j$ and $b_j$ are real coefficients, and we additionally let $W_i(x)\equiv 0$ as $|x_2|>\pi$, $i=1,2$. The corresponding eigenvalue problem  is  approximated using  a  second-difference numerical scheme with Dirichlet boundary conditions
at $x_1=0$ and $x_1=\pi$ and a decay condition
at $x_2\to\pm \infty$. In order to achieve  a  numerically efficient approximation of decay condition
at $x_2\to\pm \infty$, a quasi-equidistant grid \cite{kalitkin}  is  used with  a  step size gradually increasing towards  $x_2\to\pm \infty$. For small $\e\lessapprox0.2$, where the  localization of  the  eigenfunctions in $x_2$-direction is extremely weak, and an adequate approximation of   the  decay condition as $x_2\to\pm \infty$ is practically impossible, we use the  Neumann condition in order to approximate slowly decaying oscillating tails of   the  eigenfunctions:
\begin{equation*}
\partial_{x_2} \Psi(x_1, \pm X_0)=0,\quad\text{where}\quad X_0 \gg 1.
\end{equation*}

\begin{figure}
	\begin{center}
		 \includegraphics[width=0.68\columnwidth]{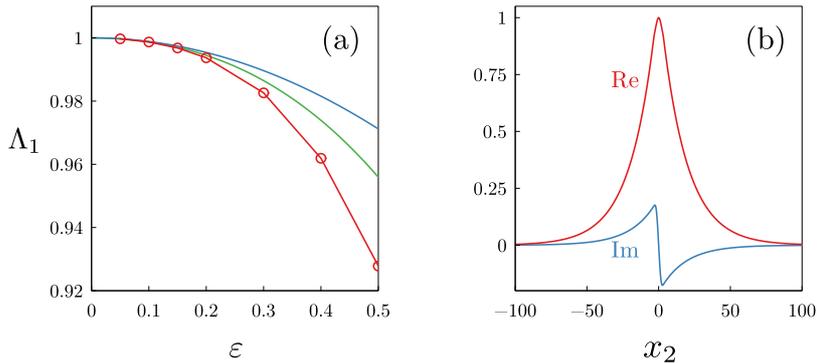}
	\end{center}
	\caption{\small (a)  Eigenvalue $\Lambda_1(\e)$ emerging from the bottom of the spectrum
computed using only the first term in expansion (\ref{1.4}) [blue curve] and two terms in (\ref{1.4}) [green curve]. Red points connected  by
red lines are obtained from direct numerical solution of the eigenvalue problem. (b)  The real  and imaginary parts of the eigenfunction at $\e=0.1$  plotted as a function of $x_2$ for $x_1=\pi/2$. Here $a_1=1$, $a_3=4$,  $b_1=0.5$ and all other coefficients are zero.}
		\label{fig:p=1}
	\end{figure}

\begin{figure}
	\begin{center}
		\includegraphics[width=0.68\columnwidth]{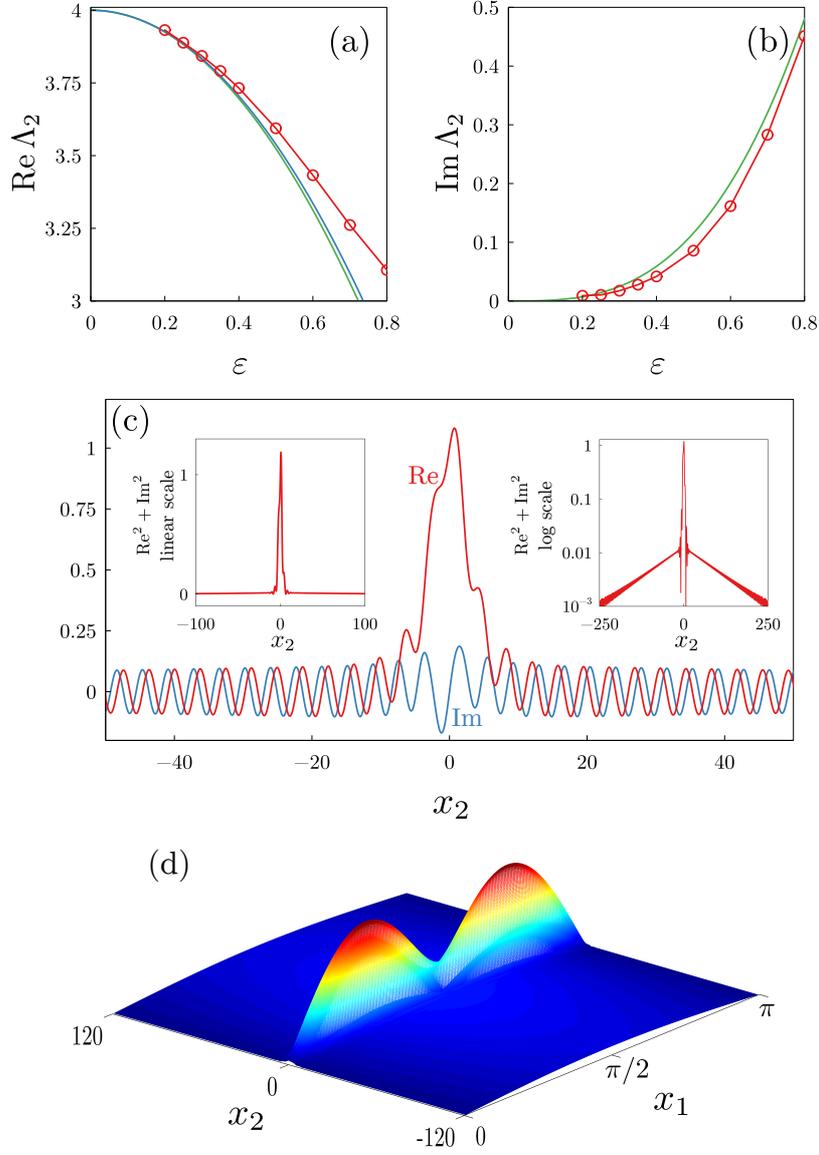}
	\end{center}
	\caption{\small  The real  and imaginary parts of the eigenvalue $\Lambda_2(\e)$  	emerging from an internal threshold in
the essential spectrum	computed using only the first term in expansion (\ref{1.9}) [blue curve] and two terms in (\ref{1.9}) [green curve]. Red points connected  by
red lines are obtained from direct numerical evaluation of the spectrum. (c)  The real  and imaginary parts of the eigenfunction at $\e=0.3$ plotted as a function of $x_2$ for $x_1=\pi/4$. {Two insets  show the plots of   the squared amplitude of the eigenfunction $|\Psi|^2 = (\mathrm{Re\,}\Psi)^2 + (\mathrm{Im\,}\Psi)^2$, using larger domains in the horizontal axes and linear (left inset) and logarithmic  (right inset) scales in the vertical axes. }   (d) Full plot of the modulus of the   eigenfunction.  In all panels $a_1=1$, $b_2=3$, and all other coefficients $a_i$ and $b_i$ are zero.  Only the eigenvalue with a positive imaginary part is shown. There also exists a complex-conjugate eigenvalue with a negative imaginary part and a $\PT$-conjugate eigenfunction. }
	\label{fig:p=2}
\end{figure}

In Fig.~\ref{fig:p=1} we plot the dependencies $\Lambda_1(\e)$ obtained from the asymptotic expansions, when only the leading term is taken into account and both terms are used, for the particular set of parameters $a_1=1$, $b_1=0.1$, $a_2=b_2=0$. For   $\e\lessapprox 0.2$,  the agreement between the analytical predictions and  numerical results is rather good, while for large values of $\e$ it is only qualitative.   In Fig.~\ref{fig:p=2} we  plot the same data for two eigenvalues bifurcating from the internal threshold with $p=2$; here only  the eigenvalue with positive imaginary part is shown.  The following set of   parameters is used: $a_1=1$, $b_1=0$, $a_2=0$, and $b_2=3$.  Again, the numerical results are in a good agreement with asymptotic expansions for weak perturbations, and in a qualitative agreement for stronger ones. Real part of the eigenvalue is found to decrease with the increase of $\varepsilon$. Nevertheless, for all values of $\varepsilon$ shown in Fig.~\ref{fig:p=2} the real part of eigenvalue $\Lambda_2(\varepsilon)$ is larger than the lower edge of the essential spectrum. Therefore, in the optical context, the corresponding eigenfunction,
 whose three-dimensional modulus plot is shown in Fig.~\ref{fig:p=2}(d), indeed represents a non-Hermitian generalization of a bound state in the continuum.

\subsubsection{Two-dimensional Bose-Einstein condensate with parabolic trapping} In the previous example a choice of $\om$ was not really important and the discussed results are of a more general nature. For instance, we can choose the operator $\Op'$ being a quantum harmonic oscillator. Namely, let $\om=\mathds{R}$ and
\begin{equation*}
\Op':=-\frac{d^2\ }{dx_1^2}+x_1^2 \quad \text{on}\quad \mathds{R}.
\end{equation*}
Then $\Om=\mathds{R}^2$ and the operator $\Op$ becomes
\begin{equation*}
\Op=-\D+x_1^2  \quad \text{on}\quad \mathds{R}^2.
\end{equation*}
As a perturbation, we choose $\cL(\e)=V_1$, where $V_1$ is a complex-valued compactly supported potential on $\mathds{R}^2$. Here
\begin{equation*}
\L_p=2p+ {1},\qquad \psi_p(x_1)=\frac{e^{-x^2/2}}{\sqrt{2^p p!\sqrt{\pi}}}
H_p\left(x_1\right),\qquad H_n(t):=(-1)^n e^{t^2} \frac{d^n\ }{dt^n}e^{-t^2},
\end{equation*}
i.e., $H_p$ are Hermite polynomials, and $p=0,1,\ldots$.
Then all calculations and results from previous example can be easily reproduced, just in all formulae   the functions $\frac{\sqrt{2}}{\sqrt{\pi}}\sin j x_1$ are to be replaced by the functions $\psi_j(x_1)$ introduced above.

The described situation corresponds
to a
two-dimensional cloud of  Bose-Einstein condensate with a parabolic confinement in $x_1$ direction. Assuming that the interparticle interactions are negligible, that is, the condensate is effectively linear, we can model its dynamics by the Schr\"oginer-like  equation, which in the theory of Bose-Einstein condensates is known as Gross-Pitaevskii equation \cite{Pitaevskii}:
\begin{equation}
\iu\partial_t \Phi + \D
\Phi - x_1^2\Phi - \e V_\e(x)\Phi = 0,
\end{equation}
where $\Phi(x_1, x_2, t)$ stands for
the macroscopic wavefunction of the condensate. Again, for stationary states in the form   $\Phi = e^{-\iu\lambda   t} \Psi$, where $\lambda$ has the meaning of the chemical potential, the problem is reduced to the described eigenvalue problem. Positive and negative imaginary parts of the perturbation correspond to the injecting the particles from an external source and absorption of the particles, respectively.

\subsubsection{Three-dimensional circular waveguide and three-dimensional  Bose-Einstein condensate }

Here we consider two examples of a three-dimensional waveguide and a three-dimensional quantum oscillator, which extend the above examples of the planar waveguide and the harmonic oscillator adduced above.

In first example we deal with a circular waveguide assuming that $\om=\{x'=(x_1,x_2):\, |x'|<1\}$ and the operator $\Op'$ is introduced as the Schr\"odinger operator with a radially symmetric potential subject to the Dirichlet condition:
\begin{equation*}
\Op'=-\D_{x'}+V_0\quad\text{in}\quad \om,\qquad V_0=V_0(|x'|).
\end{equation*}
Then the domain $\Om$ is a straight cylinder along the axis $x_3$ with the cross-section $\om$ and
\begin{equation*}
\Op=-\D+V_0(|x'|)\quad\text{in}\quad \Om
\end{equation*}
subject to the {homogeneous} Dirichlet condition. The operator $\Op'$ has a purely discrete spectrum and  thanks  to the assumed radial symmetricity for the potential $V_0$, it possesses double eigenvalues $\L_p=\L_{p+1}>\L_1$ such that the associated {eigenfunctions, orthonormalized in $L_2(\om)$,} read as $\psi_p(x')=\Psi(|x'|)\cos s\tht$,  $\psi_{p+1}(x')=\Psi(|x'|)\sin s\tht$, where $s$ is some fixed integer number,   $\Psi$ is some real  function, and $\tht$ is a polar angle associated with $x'$. Such eigenvalues are  degenerate internal thresholds in the essential spectrum. We define then a $\mathcal{PT}$-symmetric potential $V_1$   by formula (\ref{1.23}) and in addition, we assume that $W_1$ is even in $x_1$. Then it is easy to see that the corresponding matrix $\rM_1$ is diagonal:
\begin{equation*}
\rM_1=
\begin{pmatrix}
\mu_1 & 0
\\
0 & \mu_2
\end{pmatrix},\qquad \mu_1:=-\frac{1}{2}\int\limits_{\Om} W_1\psi_p^2\di x,\qquad \mu_2:=-\frac{1}{2}\int\limits_{\Om} W_1\psi_{p+1}^2\di x.
\end{equation*}
Then we assume that $\mu_1\ne\mu_2$ and $\mu_1<0$, $\mu_2<0$;  these conditions can be easily satisfied by choosing appropriately $W_1$. Since for both $\mu_i$ we have $q_i=1$, for each corresponding  pole we can apply asymptotic expansion (\ref{1.2}):
\begin{equation*}
k_{i,\tau}(\e)=\e\mu_i +\frac{\e^2}{2}\int\limits_{\Om}  V_1(x) U_{i,\tau}(x)\psi_{i-1+p} \di x
+O(\e^3),
\end{equation*}
where $U_{i,\tau}:=\cG_{1,\tau} V_1\psi_{i-1+p}$ are solutions to the boundary value problem
\begin{equation*}
(-\D-\L_p)U_{i,\tau}=V_1 \psi_{i-1+p}-\psi_{i-1+p}
\int\limits_{\Om} V_1(t',x_3)\psi_p^2(t') \di t'{dx_3} \quad\text{in}\quad\Om,\qquad U_{i,\tau}=0\quad\text{on}\quad\p\Om,
\end{equation*}
and can be found by a separation of variables similar to (\ref{1.8}). Then we can reproduce calculations from the first example and obtain that choosing appropriately the function $W_2$, we can satisfy conditions (\ref{2.26a}), (\ref{2.27b}) for all poles $k_{i,\tau}$, $i=1,2$, $\tau\in\{-1,+1\}$ and this means that these poles correspond to eigenvalues. In other words, this means that for appropriately chosen $W_1$ and $W_2$, we can generate four different simple eigenvalues of the operator $\Op_\e$ in the vicinity of the double eigenvalue $\L_p$ of the operator $\Op'$ serving as an internal threshold in the essential spectrum.

A similar situation can be realized also in other three-dimensional models. For instance, we can let $\om=\mathds{R}^2$ and $\Op'=-\D_{x'}+V_0(|x'|)$, where  $V_0=V_0(t)$  is some function growing unboundedly at infinity. Such operator $\Op'$ again can have double eigenvalues with the eigenfunctions of form $\Psi(|x'|)\cos s\tht$ and  $\Psi(|x'|)\sin s\tht$. Linear combinations $\Psi(|x'|)\cos s\tht \pm \iu  \Psi(|x'|)\sin s\tht$ correspond to vortex states with integer $s$ being the vorticity or topological charge. Therefore, in the context of Bose-Einstein condensates, this mechanism can be potentially applied for generation of localized in all three spatial dimensions vortex rings (see e.g. \cite{VR}).

We also observe that in both discussed three-dimensional examples the operator $\Op'$ can have not {only} double eigenvalues, but also ones of higher multiplicities $n$. And in such cases, it is possible to find a $\mathcal{PT}$-symmetric potential $V_1$ generating $2n$ eigenvalues of the operator $\Op_\e$ in the vicinity of the considered multiple eigenvalue.

\section{{Meromorphic} 
continuation}
\label{sec:proof1}

In this section we prove Theorem~\ref{thAnCo} on the  {meromorphic} 
continuation of the resolvent of the operator $\Op_\e$.  First we prove some auxiliary statements in a separate subsection and then we prove the theorem.

\subsection{Auxiliary lemmata}\label{ssAuLms}
In this subsection we prove three auxiliary lemmata, which will be employed in the proof of Theorem~\ref{thAnCo}.
The first statement is Lemma~\ref{lm2.1}.

\begin{proof}[Proof of Lemma~\ref{lm2.1}]
Reproducing literally the proof of Lemma~2.3 in \cite{Izv11}, one can check easily the identity $\essspec(\Op_\e)=\essspec(\Op)$. The operator $\Op$ can be represented as a sum of tensor products $\Op=\Op'\otimes \cI + \cI\otimes\Op_0$. And since
\begin{equation*}
\essspec(\Op_0)=\spec(\Op_0)=[0,+\infty),\qquad \inf\spec(\Op')=\L_1,
\end{equation*}
we immediately get:
\begin{equation*}
\spec(\Op)=\essspec(\Op)
=[\L_1,+\infty).
\end{equation*}
This completes the proof.
\end{proof}

The next statement is Lemma~\ref{lmLiRes}.

\begin{proof}[Proof of Lemma~\ref{lmLiRes}]
For each $u\in \Dom(\Op)\cap L^\bot$, each $\psi_j$ and almost each $x_d\in\mathds{R}$ we have:
\begin{align*}
\big((\Op^\bot u)(\cdot,x_d),\psi_j\big)_{L_2(\om)}=&\hf'\big(u(\cdot,x_d),\psi_j\big) - \left(\frac{\p^2 u}{\p x_d^2}(\cdot,x_d),\psi_j\right)_{L_2(\om)}
\\
=& -\L_j\big(u(\cdot,x_d),\psi_j\big)_{L_2(\om)} -\frac{d^2\ }{dx_d^2}\big(u(\cdot,x_d),\psi_j\big)_{L_2(\om)}=0.
\end{align*}
Hence, the operator $\Op^\bot$ maps $\Dom(\Op)\cap L^\bot$ into $L^\bot$. This  is an unbounded self-adjoint operator in $L^\bot$
associated with the restriction of the form $\hf$  on $\Dom(\hf)\cap L^\bot$. Moreover, for each $u\in\Dom(\hf)\cap L^\bot$ we have:
\begin{align*}
\hf(u,u)=& \int\limits_{\mathds{R}} \hf'\big(u(\cdot,x_d),u(\cdot,x_d)\big)\,dx_d + \int\limits_{\Om} \left|\frac{\p u}{\p x_d}\right|^2 \,dx \geqslant \int\limits_{\mathds{R}} \hf'\big(u(\cdot,x_d),u(\cdot,x_d)\big)\,dx_d
\\
\geqslant & c_0 \int\limits_{\mathds{R}^{d}} \|u(\cdot,x_d)\|_{L_2(\om)}^2\,dx_d=c_0\|u\|_{L_2(\Om)}^2.
\end{align*}
Hence, the spectrum of the operator $\Op^\bot$ is located in $[c_0,+\infty)$. The proof is complete.
\end{proof}

 The third lemma provides  {a meromorphic} 
 continuation for the resolvent of the unperturbed operator.

\begin{lemma}\label{lmLiCo}
Fix $p\in\{1,\ldots,m\}$ and  $\tau\in\{-1,+1\}$ and
let $\L_p=\ldots=\L_{p+n-1}$ be an  $n$-multiple eigenvalue of the operator $\Op'$, where $n\geqslant1$.
There exists a sufficiently small fixed $\d>0$ such that for all complex  $k\in B_\d$ and all  $f\in L_2(\Om,e^{{\vt}|x_d|}dx)$ the boundary value problem
\begin{equation}\label{3.2}
 \begin{gathered}
\left(
 -\sum\limits_{i,j=1}^{d-1}\frac{\p\ }{\p x_i} A_{ij} \frac{\p\ }{\p x_j} + \iu \sum\limits_{j=1}^{d-1} \left(A_j\frac{\p\ }{\p x_j} +\frac{\p\ }{\p x_j}A_j \right)+A_0 -\L_p+k^2
 \right)u=f\quad\text{in}\quad\Om,
 \\
\cB u=0 \quad\text{on}\quad\p\Om,
\end{gathered}
\end{equation}
is solvable in $W_{2,loc}^2(\Om)$ and possesses a solution, which can be represented as
\begin{equation}\label{3.3a}
u=\cA_{1,\tau}(k)f,
\end{equation}
where
$\cA_{1,\tau}$ is a linear operator mapping $L_2(\Om,e^{{\vt}|x_d|}dx)$  into $W_2^2(\Om,e^{-{\vt}|x_d|}dx)$. This operator is bounded and meromorphic in $k\in B_\d$. It has the only pole in $B_\d$, which is at zero and simple:
\begin{align}\label{3.7}
&\cA_{1,\tau}(k)=\frac{1}{k}\cA_2+\cG_{p,\tau} +k\cA_{4,\tau}(k),
\\
&\cA_2 f:=
\sum\limits_{j=p}^{p+n-1}\psi_j\ell_j f,\qquad
\ell_j f:=\frac{1}{2}\int\limits_{\Om}\overline{\psi_j(x')}f(x)\di x,\label{3.8}
\end{align}
where $\cA_{4,\tau}(k):\,L_2(\Om,e^{{\vt}|x_d|}dx)\to W_2^2(\Om,e^{-{\vt}|x_d|}dx)$ is an operator bounded uniformly in $k\in\overline{B_\d}$ and holomorphic in $k\in B_\d$, {and, we recall, $\cG_{p,\tau}$ is the operator defined in (\ref{3.9}).}

For sufficiently large $x_d$, the solution $u$ given by (\ref{3.3a}) can be represented as
\begin{equation}\label{3.10}
u(x,k)=\sum\limits_{j=1}^{m}u_j^\pm(x_d,k)\psi_j(x_d)+u_\bot^\pm(x,k),\qquad \pm x_d>R,
\end{equation}
where $R$ is some fixed number,
$u_j^\pm\in L_2(I_\pm,e^{{\mp\vt} x_d}dx_d)$ are some meromorphic in $k\in B_\d$
functions possessing the asymptotic behavior
\begin{equation}\label{3.10a}
\begin{aligned}
&u_j^\pm(x_d)= e^{-\tau K_j(k)|x_d|}\big(C_j^\pm(k)+O({e}^{-{\tilde{\vt}}|x_d|})\big),\quad && |x_d|\to\infty,\qquad j=1,\ldots,p- 1,
\\
&u_j^\pm(x_d)=e^{- K_j(k)|x_d|}\big(C_j^\pm(k)+O({e}^{-{\tilde{\vt}}|x_d|})\big), && |x_d|\to\infty,\qquad j=p,\ldots,m,
\end{aligned}
\end{equation}
$C_j^\pm(k)$ are some meromorphic in $k\in B_\d$ functions, $0<{\tilde{\vt}}<{\vt}$ is some fixed constant independent of $k$ and $x$,
and $u_\bot^\pm\in W_2^2(\Om_R^\pm)$ are some functions meromorphic in $k\in B_\d$ and obeying the identities
\begin{equation}\label{3.10b}
(u_\bot^\pm(\cdot,x_d),\psi_j)_{L_2(\Om)}=0
\end{equation}
for almost each $x_d\in I_\pm$ and for each $j=1,\ldots,m$.
\end{lemma}

\begin{proof}
The functions
\begin{equation*}
f_j(x_d):=\big(f(\cdot,x_d),\psi_j\big)_{L_2(\om)}
\end{equation*}
obviously belong to $L_2(\mathds{R},e^{{\vt}|x_d|}dx_d)$
since $f\in L_2(\Om,e^{{\vt}|x_d|}dx)$. It is straightforward to check that the functions
\begin{equation}\label{3.5}
\begin{aligned}
&u_j(x_d,k):=\frac{1}{2 \tau K_j(k)} \int\limits_{\mathds{R}} e^{-\tau K_j(k)|x_d-y_d|} f_j(y_d)\,dy_d\quad &&\text{as}\quad j<p-1,
\\
&u_j(x_d,k):=\frac{1}{2  K_j(k)} \int\limits_{\mathds{R}} e^{- K_j(k)|x_d-y_d|} f_j(y_d)\,dy_d\quad &&\text{as}\quad j\geqslant p,
\end{aligned}
\end{equation}
 solve the equations
\begin{equation*}
-u_j'' +(\L_j-\L_p+k^2)u_j=f_j\quad\text{in}\quad \mathds{R},\qquad j=1,\ldots,m.
\end{equation*}
Employing these facts, we seek a solution to problem (\ref{3.2}) as
\begin{equation}\label{3.1}
u(x,k)=\sum\limits_{j=1}^{m} u_j(x_d',k)\psi_j(x') + u_\bot(x,k)
\end{equation}
and for $u_\bot$ we immediately get problem (\ref{3.2}) with $f$ replaced by the function $f^\bot$ defined in (\ref{3.11}).
We see easily that $f^\bot\in L^\bot$.

Thanks to Lemma~\ref{lmLiRes}, the resolvent $(\Op^\bot-\L_p+k^2)^{-1}$ is well-defined for all sufficiently small complex $k\in B_\d$ provided $\d$ is small enough. This resolvent is holomorphic in $k\in B_\d$ as an operator in $L_2(\Om)$ and is uniformly bounded in $k\in \overline{B_\d}$. It can be expanded via the standard Neumann series:
\begin{equation}\label{3.12}
(\Op^\bot-\L_p+k^2)^{-1}=\sum\limits_{j=0}^{\infty}(-k^2)^j\big((\Op^\bot-\L_p)^{-1}
\big)^{j+1}.
\end{equation}
This expansion implies that the resolvent  $(\Op^\bot-\L_p+k^2)^{-1}$ is also holomorphic as an operator from $L_2(\Om)$ into $W_2^2(\Om)$. We define then
 \begin{equation*}
 u_\bot:=(\Op^\bot-\L_p+k^2)^{-1}f^\bot
\end{equation*}
and this obviously gives a solution to problem~(\ref{3.2}). We denote the operator mapping $f$ into the described solution by $\cA_{1,\tau}(k)$ and let us show that it possesses all stated properties.

First of all we observe that the introduced operator is independent of the choice of $\tau$ as $\L_p=\L_1$ since in this case the functions $u_j$ with $\L_j<\L_p$ are missing and the above constructions become independent of $\tau$.

Just by the embeddings $L_2(\Om,e^{{\vt}|x_d|}dx)\subset L_2(\Om)$ and $W_2^2(\Om)\subset W_2^2(\Om,e^{-{\vt}|x_d|}dx)$, the resolvent $(\Op^\bot-\L_p+k^2)^{-1}$  is a bounded operator from $L_2(\Om,e^{{\vt}|x_d|}dx)$ into $W_2^2(\Om,e^{-{\vt}|x_d|}dx)$ holomorphic in $k\in B_\d$ for sufficiently small $\d$ and bounded uniformly in $k\in\overline{B_\d}$. An operator mapping $f$ into $\sum\limits_{j=1}^{m} u_j(x_d,k)\psi_j(x')$ is given explicitly and by straightforward calculations we can check that this is also a bounded operator from $L_2(\Om,e^{{\vt}|x_d|}dx)$ into $W_2^2(\Om,e^{-{\vt}|x_d|}dx)$   holomorphic in $k\in B_\d$. The calculations are based on estimates of the following kind:
\begin{align*}
\|u_j\psi_j\|_{L_2(\Om,e^{-{\vt}|x_d|}dx)}^2=&\|u_j\|_{L_2(\mathds{R},e^{-{\vt}|x_d|}dx)}^2
\\
\leqslant & \frac{\|f_j\|_{L_2(\mathds{R},e^{{\vt}|x_d|}dx_d)}^2}{2|K_j|^2}\int\limits_{\mathds{R}^2}  e^{-|x_d-y_d|\tau\RE K_j(k)-a(|x_d|+|y_d|)}\di x_d \di y_d
\\
\leqslant & C(k)\|f\|_{L_2(\Om,e^{{\vt}|x_d|}dx)}^2,
\end{align*}
where $C(k)$ is some constant independent of $f$.
{Since the functions $f$ belong to $L_2(\Om,e^{\vt|x_d|}dx_d)$,} the functions $u_j$ can be expanded into the {Laurent} 
series with respect to the small parameter $k${; in particular, this means that the functionals $\ell_j: L_2(\Om,e^{\vt|x_d|}dx)\to\mathds{C}$ are well-defined and bounded.}
Substituting these expansions and (\ref{3.12}) into (\ref{3.1}), we immediately get representation (\ref{3.7}), (\ref{3.8}), (\ref{3.9}).

In view of formula (\ref{3.1}), in order to prove representation (\ref{3.10}), (\ref{3.10a}), (\ref{3.10b}), it is sufficient to analyze the behavior of the functions $u_j$ at infinity.  This can be done easily by  an  obvious identity
\begin{align}
&u_j(x_d,k)-\frac{e^{\mp \tau K_j(k)x_d}}{2 \tau K_j(k)} \int\limits_{\mathds{R}}e^{\pm \tau K_j(k)y_d} f_j(y_d)\,dy_d=\tilde{u}_j^\pm(x_d,k),\qquad \pm x_d>0,\label{3.13}
\\
&\tilde{u}_j^\pm(x_d,k):= \frac{1}{2 K_j(k)} \int\limits_{x_d}^{\pm\infty} \Big(e^{K_j(k)(x_d-y_d)}- e^{-K_j(k)(x_d-y_d)}\Big)f_j(y_d)\,dy_d, \nonumber
\end{align}
and an  estimate
\begin{equation}\label{3.15}
\begin{aligned}
|\tilde{u}_j^\pm(x_d,k)|\leqslant &\frac{1}{2 |K_j(k)|}\left(
\int\limits_{x_d}^{\pm\infty} \Big|e^{K_j(k)(x_d-y_d)}- e^{-K_j(k)(x_d-y_d)}\Big|^2e^{-a|y_d|} \,dy_d,
\right)^\frac{1}{2}\|f_j\|_{L_2(\mathds{R},e^{{\vt}|x_d|}dx_d)}
\\
\leqslant & C(k)e^{-{\tilde{\vt}}|y_d|}\|f_j\|_{L_2(\mathds{R},e^{{\vt}|x_d|}dx_d)},\qquad \pm x_d>0,
\end{aligned}
\end{equation}
where $C(k)>0$ and $0<{\tilde{\vt}<\vt}$ are some constants independent of $x_d$ and $f$.
The proof is complete.
\end{proof}

\subsection{
{Meromorphic} continuation for the resolvent of the perturbed operator}

This section is devoted to the proof of Theorem~\ref{thAnCo}.
We begin with rewriting the operator equation
\begin{equation*}
(\Op_\e-\L_p+k^2)u_\e=f,
\end{equation*}
as a boundary value problem (\ref{2.5}). Then we denote
\begin{equation*}
g_\e:=f-\e\cL(\e) u_\e
\end{equation*}
and  rewrite
(\ref{2.5})  as problem (\ref{3.2}) with the function $f$ replaced by $g_\e$. According Lemma~\ref{lmLiCo}, such problem is solvable and there exists a solution given by formula (\ref{3.3a}):
\begin{equation}\label{4.1}
u_\e=\cA_{1,\tau}(k)g_\e.
\end{equation}
We substitute this formula into corresponding boundary value problem (\ref{2.5}) and this leads to an operator equation in the space $L_2(\Om,e^{{\vt}|x_d|}dx)$:
\begin{equation*}
g_\e+\e\cL(\e)\cA_{1,\tau}(k)g_\e=f.
\end{equation*}
Then we substitute
representation (\ref{3.7}), (\ref{3.8}), (\ref{3.9})
into this equation:
\begin{equation}\label{4.3}
g_\e + \frac{\e}{k}
\sum\limits_{j=p}^{p+n-1}\big(\ell_j g_\e\big)\cL(\e)\psi_j
+\e\cL(\e)\big(\cG_{p,\tau} +k\cA_{4,\tau}(k)\big)g_\e=f.
\end{equation}
Thanks to the properties of the operators $\cG_{p,\tau}$ and $\cA_{4,\tau}(k)$ described in Lemma~\ref{lmLiCo}, the operator $\cL(\e)\big(\cG_{p,\tau} +k\cA_{4,\tau}(k)\big)$ is bounded uniformly in $\e$  and $k$ as an operator in $L_2(\Om,e^{{\vt}|x_d|}dx)$ and is holomorphic in $k\in B_\d$. Hence, for sufficiently small $\e$, the operator $(\cI+\e\cL(\e)\big(\cG_{p,\tau} +k\cA_{4,\tau}(k)\big))$ is boundedly invertible and the inverse is holomorphic in $k\in B_\d$ {as an operator in $L_2(\Om,e^{\vt|x_d|}dx)$}; hereinafter the symbol $\cI$ stands for the identity mapping. We denote
\begin{equation}\label{4.10}
\cA_{5,\tau}^\e(k):=\big(\cI+\e\cL(\e)\big(\cG_{p,\tau} +k\cA_{4,\tau}(k)\big)\big)^{-1}
\end{equation}
and apply this operator to equation (\ref{4.3}). This gives:
\begin{equation}\label{4.4}
g_\e + \frac{\e}{k}
\sum\limits_{j=p}^{p+n-1}\big(\ell_j g_\e\big)\cA_{5,\tau}^\e(k)\cL(\e)\psi_j=\cA_{5,\tau}^\e(k)f.
\end{equation}
Our next step is to find $\ell_j g_\e$. Once we do this, we shall be able to find $g_\e$ from equation (\ref{4.3}) as
\begin{equation}\label{4.5}
g_\e= -\frac{\e}{k}
\sum\limits_{j=p}^{p+n-1}\big(\ell_j g_\e\big)\cA_{5,\tau}^\e(k)\cL(\e)\psi_j+\cA_{5,\tau}^\e(k)f
\end{equation}
and recover the function $u_\e$ by formula (\ref{4.1}).

We apply the functionals $\ell_i$ to equation (\ref{4.4}) with all $i$ such that $\L_i=\L_p$. As a result, we arrive to a system of linear equations for the vector  $\ell g_\e:=(\ell_j g_\e)_{j=p,\ldots, p+n-1}$; this is a vector-column. The system reads as
\begin{equation}\label{4.6}
\left(\rE+\frac{\e}{k}\rM_{\e,\tau}(k)\right)\ell g_\e=F_\e(k).
\end{equation}
{Here} $\rE$ is the unit matrix and $\rM_{\e,\tau}(k)$ is a square matrix with entries $M_\e^{ij}(k):=\ell_i \cA_{5,\tau}^\e(k)\cL(\e)\psi_j$, where $i$ counts the rows and $j$ does the columns in the matrix $\rM_{\e,\tau}(k)${, while the} 
symbol $F_\e(k)$ denotes a vector column  with coordinates $\ell_i \cA_{5,\tau}^\e(k)f$, $i=p,\ldots,p+n-1$. {
 As it was shown in the proof of Lemma~\ref{lmLiCo}, the functionals $\ell_i: L_2(\Om,e^{\vt|x_d|}dx)\to\mathds{C}$ are bounded and since the operator $\cA_{5,\tau}^\e(k)$ is a bounded one in $L_2(\Om,e^{\vt|x_d|}dx)$, the introduced matrix $\rM_{\e,\tau}(k)$ and vector $F_\e(k)$ are well-defined. }

In view of the holomorphy of the operator $\cA_{5,\tau}^\e$, the entries of the matrix $\rM_{\e,\tau}(k)$ are holomorphic in $k$ and this implies that the matrix $\rM_{\e,\tau}(k)$  is holomorphic in $k\in B_\d$.
Hence, the determinant of the matrix $k\rE+\e\rM_{\e,\tau}(k)$ is a holomorphic in $k$ function and therefore, the matrix $(k\rE+\e\rM_{\e,\tau}(k))^{-1}$ is well-defined as a meromorphic in $k$ matrix. This allows us to solve system (\ref{4.6}):
\begin{equation}\label{4.7}
\ell g_\e=k(k\rE+\e\rM_{\e,\tau}(k))^{-1}F_\e(k)=k\big(\ell_j^\e(k)f\big)_{j=p,\ldots, p+n-1},\qquad \ell_j g_\e=k\ell_j^\e(k)f,
\end{equation}
where $\ell_j^\e(k)$ are functionals on $L_2(\Om,e^{{\vt}|x_d|}dx)$ meromorphic  in $k\in B_\d$. Substituting these formulae into (\ref{4.5}), we can find $g_\e$ and determine $u_\e$ by formula (\ref{4.1}):
\begin{equation}\label{4.8}
\begin{aligned}
u_\e=&
\sum\limits_{j=p}^{p+n-1}\big(\ell_j^\e f\big)\Big(\psi_j-\e
\big(\cG_{p,\tau}+k\cA_{4,\tau}(k)\big)\cA_{5,\tau}^\e(k)\cL(\e)\psi_j\Big) + \big(\cG_{p,\tau}+k\cA_{4,\tau}(k)\big)\cA_{5,\tau}^\e(k)f
\\
=&
\sum\limits_{j=p}^{p+n-1}\big(\ell_j^\e f\big)\big(\cI+\e\big(\cG_{p,\tau}+k\cA_{4,\tau}(k)\big) \cL(\e)\big)^{-1}\psi_j
+ \big(\cG_{p,\tau}+k\cA_{4,\tau}(k)\big)\cA_{5,\tau}^\e(k)f=:\cR_{\e,\tau}(k)f.
\end{aligned}
\end{equation}
The function $u_\e$ solves problem (\ref{2.5}) and satisfies representation  (\ref{2.6}), (\ref{2.6a}), (\ref{2.6b}). The latter are implied by similar representations (\ref{3.10}), (\ref{3.10a}), (\ref{3.10b}) and formula (\ref{4.1}). Since for sufficiently small $k$ the functions $K_j$, $\L_j<\L_p$, satisfy  asymptotic identities (\ref{4.11}),
the functions $e^{-\tau K_j(k)|x_d|}$, $\L_j<\L_p$ and $e^{-k|x_d|}$ decay exponentially at infinity as $\RE k>0$ and $\tau\IM k^2<0$. Hence, under the same conditions, we have $u_\e=(\Op_\e-\L_p+k^2)^{-1}f$ and this means that the operator $\cR_{\e,\tau}$ provides 
 {a meromorphic} continuation of the resolvent $(\Op_\e-\L_p+k^2)^{-1}$.

We   observe that the introduced operator $\cR_{\e,\tau}$ is meromorphic in $k\in B_\d$ as an operator from $L_2(\Om,e^{{\vt}|x_d|}\,dx)$ into $W_2^2(\Om)\cap L_2(\Om,e^{-{\vt}|x_d|}\,dx)$ due to formulae (\ref{4.1}), (\ref{4.5}), (\ref{4.7}) and representation (\ref{3.7}), (\ref{3.8}), (\ref{3.9}).

According {to} formula (\ref{4.8}), the poles of the  operator $\cR_{\e,\tau}(k)$ defined in this formula coincide with the poles of the functionals $\ell_j^\e(k)$. Let at least one of these functionals have a pole at a point $k_\e\in B_\d$. In view of the definition of the functionals in (\ref{4.7}), this means that the matrix $(k_\e\rE+\e\rM_{\e,\tau}(k_\e))$ is degenerate and by the Cramer's rule, system (\ref{4.6}) with $F_\e(k)=0$, $k=k_\e$
has a non-trivial solution, which we denote by $l^\e=(l^\e_j)_{j:\, \L_j=\L_p}$. We observe that by the definition,
\begin{equation}\label{4.12}
k_\e l_i^\e=-\e\sum\limits_{j=p}^{p+n-1} l_j^\e \ell_i \cA_{5,\tau}^\e(k)\cL(\e)\psi_j.
\end{equation}
Employing these identities and (\ref{3.7}), (\ref{3.8}), it is straightforward to check that the formula
\begin{equation}\label{4.9a}
\begin{aligned}
\psi_\e:=&-\e\lim\limits_{k\to k_\e} \cA_{1,\tau}(k)
\sum\limits_{i=1}^{n}  l_i^\e \cA_{5,\tau}^\e(k)\cL(\e)\psi_{i+p-1}
\\
=&
\sum\limits_{i=1}^{n} l_i^\e \psi_{i+p-1} -\e
\big(\cG_{p,\tau}+k_\e\cA_{4,\tau}(k_\e)\big)\sum\limits_{i=1}^{n}  l_i^\e \cA_{5,\tau}^\e(k_\e)\cL(\e)\psi_{i+p-1}
\\
=&
\sum\limits_{i=1}^{n} l_i^\e\big(\cI+\e\big(\cG_{p,\tau}+k_\e\cA_{4,\tau}(k_\e)\big) \cL(\e)\big)^{-1}\psi_{i+p-1}
\end{aligned}
\end{equation}
  defines a non-trivial solution to problem (\ref{2.5}) with $f=0$, $k=k_\e$. Thanks to formulae (\ref{3.1}), (\ref{3.5}), (\ref{3.15}), this function can be also represented as
\begin{equation}\label{4.9b}
\begin{aligned}
&
\psi_\e=\sum\limits_{j=1}^{m} \psi_j(x')\vp_j(x_d)
+(\Op^\bot-\L_p+k_\e^2)^{-1}f^\bot,
\\
&\vp_j(x_d):=\frac{1}{2 \tau K_j(k_\e)} \int\limits_{\Om} e^{-\tau K_j(k_\e)|x_d-y_d|} f(y)\overline{\psi_j(y')}\,dy &&\text{as}\quad j<p-1,
\\
&\vp_j(x_d):=\frac{1}{2  K_j(k_\e)} \int\limits_{\Om} e^{- K_j(k_\e)|x_d-y_d|} f(y)\overline{\psi_j(y')}\,dy  &&\text{as}\quad j\geqslant p+n-1,
\\
&\vp_j(x_d):=\frac{1}{2   k_\e } \int\limits_{\Om} e^{- k_\e|x_d-y_d|} f(y)\overline{\psi_j(y')}\,dy  &&\text{as}\quad j=p,\ldots, p+n-1,\quad k_\e\ne0,
\\
&\vp_j(x_d):=-\frac{1}{2} \int\limits_{\Om} |x_d-y_d| f(y)\overline{\psi_j(y')}\,dy  &&\text{as}\quad j=p,\ldots, p+n-1,\quad k_\e=0,
\end{aligned}
\end{equation}
with
\begin{equation*}
f:=-\e\sum\limits_{i=1}^{n}  l_i^\e \cA_{5,\tau}^\e(k_\e)\cL(\e)\psi_{i+p-1}
\end{equation*}
and $f^\bot$ defined by (\ref{3.11}). We also observe that as $k_\e=0$, it follows from (\ref{4.12}) that
\begin{equation*}
(f(\cdot,x_d),\psi_j)_{L_2(\Om)}=0
\end{equation*}
for almost each $x_d\in I_\pm$ and for each $j=1,\ldots,m$.
Representations (\ref{4.9a}), (\ref{4.9b}) and relations (\ref{3.13}), (\ref{3.15}) imply representation (\ref{2.8}), (\ref{2.8a}), (\ref{2.8b}). This completes the proof of Theorem~\ref{thAnCo}.

\section{Poles of  {meromorphic} 
continuations}\label{ss:Emer}
\label{sec:proof2}

In this section we study the asymptotic behaviour of poles of the operators $\cR_{\e,\tau}$ and we prove Theorems~\ref{thEmer1},~\ref{thEmer2}. As it was shown in the proof of Theorem~\ref{thAnCo} in the previous section, the poles of the operator $\cR_{\e,\tau}(k)$ coincides with poles of the matrix $(k\rE+\e\rM_{\e,\tau}(k))^{-1}$. The orders of the mentioned poles of the operator $\cR_{\e,\tau}(k)$ and the orders of the zeroes of $\det(k\rE+\e\rM_{\e,\tau}(k))$ obviously coincide.

\subsection{Proof of Theorem~\ref{thEmer1}}
Let us study the solvability of the equation
\begin{equation}\label{5.1}
\det(k\rE+\e\rM_{\e,\tau}(k))=0
\end{equation}
in $B_\d$. It is straightforward to check that
\begin{equation*}
\det(k\rE+\e\rM_{\e,\tau}(k))=k^n+P(k,\e),\qquad P(k,\e):=\sum\limits_{i=1}^{n}\e^i k^{n-i} P_i(k,\e),
\end{equation*}
where $P_i(k,\e)$ are some functions holomorphic in $k\in B_\d$ and uniformly bounded in $k\in\overline{B_\d}$ and sufficiently small $\e$. We fix arbitrary $\d'\in(0,\d)$ and for sufficiently small $\e$ we have an uniform estimate:
\begin{equation}\label{5.3}
|P(k,\e)|\leqslant C\e \quad\text{on}\quad\p B_{\d'}.
\end{equation}
Since $|k^n|=(\d')^n$ on $\p B_{\d'}$ and the function $k\mapsto k^n$ has the only zero in $B_\d$ of order $n$ at the origin, estimate (\ref{5.3}) allows us to {apply 
the} Rouch\'e theorem and to conclude that equation (\ref{5.1}) has exactly $n$ zeroes counting their orders and these zeroes {converge} 
to the origin as $\e\to+0$.

Our next step is to show that all zeroes of equation (\ref{5.1}) are located in a ball $B_{b\e}$ for some fixed $b$. For arbitrary $b$ and sufficiently small $\e$ such that $b\e<d$, by the definition of the function $P$ we have:
\begin{equation*}
|P(k,\e)|\leqslant C\e^n(1+b+\ldots+b^{n-1}),\qquad |k|^n=\e^n b^n \quad\text{on}\quad\p B_{b\e}.
\end{equation*}
Then we choose $b$ large enough and $\e$ small enough so that
\begin{equation*}
b^n>C (1+b+\ldots+b^{n-1}), \qquad b\e<\d,
\end{equation*}
and applying the Rouch\'e theorem once again,  we see that the zeroes of equation (\ref{5.1}) are located inside the ball $B_{b\e}$. This fact allows us to seek the zeroes of equation (\ref{5.1}) as $k=z\e$, where $|z|<b$ for all sufficiently small $\e$.

The operator $\cA_{5,\tau}^\e$ defined in (\ref{4.10}) can be expanded into the standard Neumann series and this yields:
\begin{equation}\label{5.5}
\cA_{5,\tau}^\e(k)=\cI-\e\cL(\e)\big(\cG_{p,\tau}+k\cA_{4,\tau}(k)) +\e^2(\cL(\e)\big(\cG_{p,\tau}+k\cA_{4,\tau}(k)))^2\cA_{5,\tau}^\e(k).
\end{equation}
In view of the above identity, the change $k=z\e$,  $|z|<b$, and the definition of the matrix $\rM_{{\e,\tau}}$,  we can represent the latter as
\begin{equation}\label{5.20}
\rM_{\e,\tau}(k)=- \rM_1+\e \tilde{\rM}_2(z,\e),
\end{equation}
where $\tilde{\rM}_2(z,\e)$ is some uniformly bounded matrix holomorphic in $z$ and we recall that the matrix $\rM_1$ is defined in (\ref{2.9}). {Both these matrices are well-defined since they arise as the terms in the expansion for $\rM_{\e,\tau}$.
}

Then we can rewrite equation (\ref{5.1}) as
\begin{equation}\label{5.6}
\det\big(z\rE-\rM_1+\e\tilde{\rM}_2(z,\e)\big)=0.
\end{equation}
We recall that $\mu_j$, $j=1,\ldots,N$, are different eigenvalues of the matrix $\rM_1$ of the multiplicities $q_j$ and hence,
\begin{equation*}
\det(z\rE-\rM_1)=\prod\limits_{j=1}^{N}(z-\mu_j)^{q_j}.
\end{equation*}
This allows us to represent equation (\ref{5.6}) as
\begin{equation}\label{5.17}
\prod\limits_{i=1}^{N}(z-\mu_i)^{q_i} + \e \tilde{P}(z,\e)=0,
\end{equation}
where $\tilde{P}$ is some function  holomorphic in $z$ and bounded uniformly in $z$ and $\e$:
\begin{equation}\label{5.19}
|\tilde{P}(z,\e)|\leqslant C.
\end{equation}

 By $T_r(z_0)$ we denote a ball of radius $r$ in the complex plane centered at the point $z_0$.
We introduce the balls $T_{\tilde{b}\e^{\frac{1}{q_i}}}(\mu_i)$,
where $\tilde{b}$ is some fixed number. Then for sufficiently small $\e$, each ball
$T_{\tilde{b}\e^{\frac{1}{q_i}}}(\mu_i)$ contains the only zero of the function
$z\mapsto \prod\limits_{j=1}^{N}(z-\mu_j)^{q_j}$, this zero is $\mu_i$ and the order of this zero is $q_i$. On the boundary of each ball
$T_{\tilde{b}\e^{\frac{1}{q_i}}}(\mu_i)$, for sufficiently small $\e$, an obvious estimate holds true:
\begin{equation}\label{5.18}
\left| \prod\limits_{j=1}^{N}(z-\mu_j)^{q_j}\right|\geqslant C\tilde{b}^{q_i}\e,
\end{equation}
where $C$ is some fixed constant independent of $\e$ and $\tilde{b}$. Hence, in view of estimate (\ref{5.19}), we  can choose $\tilde{b}$ large enough and $\e$ small enough and apply the Rouch\'e theorem to the left hand side of equation (\ref{5.17}). This yields that each ball
$T_{\tilde{b}\e^{\frac{1}{q_i}}}(\mu_i)$ contains exactly $q_i$ zeroes of equation (\ref{5.17}), which we denote by $z_{ij}=z_{ij}(\e)$. These zeroes satisfy the asymptotic identities
\begin{equation}\label{5.25}
z_{ij}(\e)=\mu_i+O(\e^{\frac{1}{q_i}}),\qquad j=1,\ldots,q_i.
\end{equation}
Returning back to equation (\ref{5.1}), we see that it has exactly $n$ zeroes counting their orders and these zeroes obey asymptotic expansion (\ref{2.10}). The proof is complete.

\subsection{Proof of Theorem~\ref{thEmer2}}
Identity (\ref{5.5}) allows us to specify the form of the matrix $\rM_{\e,\tau}$ in more details than in (\ref{5.20}). Namely, this identity implies that
\begin{equation*}
\rM_{\e,\tau}(k)=-\rM_1+\e \rM_2^\tau+\e^2\tilde{\rM}_3(z,\e),
\end{equation*}
where $\tilde{\rM}_3(z,\e)$ is some uniformly bounded matrix holomorphic in $z$ and we recall that the matrix   $\rM_2$ is  defined in (\ref{2.11}). {All these matrices are well-defined  as the terms in the expansion for $\rM_{\e,\tau}$.}  Equation  (\ref{5.6}) is replaced by
a more detailed one:
\begin{equation}\label{5.21}
\det\big(z\rE-\rM_1+\e \rM_2^\tau+\e^2\tilde{\rM}_3(z,\e)\big)=0.
\end{equation}

We fix $i\in\{1,\ldots,N\}$ and we are going to study the asymptotic behavior of the zeroes $z_{ij}(\e)$, $j=1,\ldots,q_i$.
Let $\rS$ be a matrix reducing the matrix $\rM_1$ to its Jordan canonical form, which we denote by $\rJ:=\rS\rM_1\rS^{-1}$. Then equation (\ref{5.21}) can be rewritten as
\begin{equation*}
\det\big(z\rE-\rJ +\e \rS\rM_2^\tau\rS^{-1} +\e^2\rS\tilde{\rM}_3(z,\e)\rS^{-1}\big)=0.
\end{equation*}
Taking into consideration  the structure of the matrix $\rJ$, we
rewrite this equation as follows:
\begin{equation}\label{5.23}
 \prod\limits_{j=1}^{N}(z-\mu_j)^{q_j}  +\e(z-\mu_i)^{r_{i,\tau}}Y_1(z) + \e^2 Y_2(z,\e)=0,
\end{equation}
where  $Y_1$, $Y_2$ are some functions holomorphic in $z$ and bounded uniformly in $\e$ and $z$ and $r_{i,\tau}<q_i$ is some non-negative integer number. These functions arise as some polynomials in $z$, $\e$ and the entries of the matrices  $\rS\rM_2\rS^{-1} $ and $\rS\tilde{\rM}_3(z,\e)\rS^{-1}$.  It is clear that $(z-\mu_i)^{r_{i,\tau}}Y_1(z)=Q_{i,\tau}(z)$, where $Q_i$ is defined in (\ref{2.12}). And if $Y_1$ is not identically zero, then
\begin{equation}\label{5.22}
Y_1(z_i)=\g_{i,\tau}\prod\limits_{\substack{j=1
\\
j\ne i}}^{N}(\mu_i-\mu_j)^{q_j}\ne0,
\end{equation}
where $\g_{i,\tau}$ is defined in (\ref{2.16}). By $T_r(z_0)$ we denote a ball of radius $r$ in the complex plane centered at the point $z_0$.

We first consider the case, when $Y_1$ vanishes identically. Then the term $\e Y_1$ disappears in equation (\ref{5.23}). Here we consider the ball $T_{\tilde{b}\e^{\frac{2}{q_i}}}(\mu_i)$ and for sufficiently small $\e$, on the boundary of this ball we have the estimates
\begin{equation*}
\Big| \prod\limits_{j=1}^{N}(z-\mu_j)^{q_j} \Big|\geqslant C\tilde{b}^{q_i}\e^2,\qquad |\e^2 Y_2(z,\e)|\leqslant \tilde{C}\e^2,
\end{equation*}
where $C$ and $\tilde{C}$ are some fixed constants independent of $\e$ and $\tilde{b}$. Then we apply the Rouch\'e theorem proceeding as in (\ref{5.17}), (\ref{5.19}), (\ref{5.18}), (\ref{5.25}) and we arrive immediately to (\ref{2.13}).

Now {we} proceed to the case, when $Y_1$ is not identically zero. Here we divide equation (\ref{5.23}) by $ \prod\limits_{\substack{j=1
\\
j\ne i}}^{N}(z-\mu_j)^{q_j}$ and in view of (\ref{5.22}) we get:
\begin{equation}\label{5.26}
(z-\mu_i)^{q_i}  +\e(z-\mu_i)^{r_{i,\tau}}\g_{i,\tau}+\e (z-\mu_i)^{r_{i,\tau}+1}Y_3(z) + \e^2 Y_4(z,\e)=0,
\end{equation}
where $Y_3$, $Y_3$ are some functions holomorphic in $z$ in the vicinity of the point $\mu_i$ and bounded uniformly in $\e$ and $z$.

Suppose that $q_i\leqslant 2r_{i,\tau}$ and consider the ball $T_{\tilde{b}\e^{\frac{1}{r_{i,\tau}}}}(\mu_i)$. For sufficiently small $\e$,
on the boundary of this ball we have
\begin{equation}\label{5.27}
\begin{aligned}
|(z-\mu_i)^{q_i}  +\e(z-\mu_i)^{r_{i,\tau}}\g_{i,\tau}|\geqslant & |(z-\mu_i)^{q_i}| - \e|\g_{i,\tau}||(z-\mu_i)^{r_{i,\tau}}|
\\
=& \tilde{b}^{q_i}\e^{2-\frac{2r_{i,\tau}-q_i}{r_{i,\tau}}}
-\e^2|\g_{i,\tau}|\tilde{b}^{r_{i,\tau}}
\geqslant  \frac{ \tilde{b}^{q_i}}{2}\e^{2-\frac{2r_{i,\tau}-q_i}{r_{i,\tau}}}
\end{aligned}
\end{equation}
 and
\begin{equation}\label{5.28}
|\e (z-\mu_i)^{r_{i,\tau}+1}Y_3(z) + \e^2 Y_4(z,\e)|\leqslant C(\tilde{b}^{r_{i,\tau}+1}\e^{2+\frac{1}{r_{i,\tau}}
}+\e^2)<C\e^2,
\end{equation}
where $C$ is some constant independent of $\e$ and $\tilde{b}$.
Since $q_i-r_{i,\tau}\geqslant r_{i,\tau}+1$, it is also straightforward to check that for sufficiently small $\e$ and sufficiently large $\tilde{b}$, all zeroes of the function $z\mapsto (z-\mu_i)^{q_i}  +\e(z-\mu_i)^{r_{i,\tau}}\g_{i,\tau}$  belong to the ball $T_{\tilde{b}\e^{\frac{1}{r_{i,\tau}}}}(\mu_i)$.
Taking this fact and estimates (\ref{5.27}), (\ref{5.28}) into consideration, we apply the Rouch\'e theorem and   arrive at asymptotic expansion (\ref{2.17}).

Suppose that $q_i-2r_{i,\tau}\geqslant 1$. The zeroes of the function $z\mapsto (z-\mu_i)^{q_i}  +\e(z-\mu_i)^{r_{i,\tau}}\g_{i,\tau}$ can be found explicitly. The first of them is a $r_{i,\tau}$-multiple zero $z_0=\mu_i$ and $q_i-r_{i,\tau}$ simple zeroes
\begin{equation}\label{5.31}
z_j(\e):=\mu_i-\e^{\frac{1}{q_i-r_{i,\tau}}}(-\g_{i,\tau})^{\frac{1}{q_i-r_{i,\tau}}} e^{\frac{2\pi\iu}{q_i-r_{i,\tau}}j}, \qquad j=1,\ldots,q_i-r_{i,\tau}.
\end{equation}
First we consider a ball $T_{\tilde{b}\e^{\frac{1}{r_{i,\tau}}}}(\mu_i)$ and we observe that under our assumption, $q_i-r_{i,\tau}\geqslant r_{i,\tau}+1$. Hence, for sufficiently small $\e$, the ball $T_{\tilde{b}\e^{\frac{1}{r_{i,\tau}}}}(\mu_i)$ does not contain the zeroes $z_j(\e)$. Moreover, on the boundary of this ball, the estimate holds  true:
 \begin{align*}
|(z-\mu_i)^{q_i}  +\e(z-\mu_i)^{r_{i,\tau}}\g_{i,\tau}|\geqslant & \e|\g_{i,\tau}||(z-\mu_i)^{r_{i,\tau}}|-|(z-\mu_i)^{q_i}|
\\
=& \tilde{b}^{r_{i,\tau}}\e^2|\g_{i,\tau}|-\tilde{b}^{q_i}\e^{2+\frac{q_i-2r_{i,\tau}}{r_{i,\tau}}}> \frac{|\g_{i,\tau}|}{2}\tilde{b}^{r_{i,\tau}}\e^2.
\end{align*}
We also observe that estimate (\ref{5.28}) remains true in the considered case.  Then we apply Rouch\'e theorem to the ball  $T_{\tilde{b}\e^{\frac{1}{r_{i,\tau}}}}(\mu_i)$ and conclude that this ball contains exactly $r_{i,\tau}$ zeroes of equation (\ref{5.26}) counting their orders. This implies asymptotic expansion (\ref{2.18}).

Now we consider another ball $T_{\tilde{b}\e^{\frac{2}{q_i-r_{i,\tau}}}}(z_j(\e))$ for $z_j$ defined in  (\ref{5.31}) with some fixed $j$. For sufficiently small $\e$, this ball does not contain the origin and other points $z_s$ with $s\ne j$. On the boundary of this ball, we have the following estimate:
\begin{equation}\label{5.32}
\begin{aligned}
|(z-\mu_i)^{q_i}  +\e(z-\mu_i)^{r_{i,\tau}}\g_{i,\tau}|\geqslant & C\e^{\frac{r_{i,\tau}}{q_i-r_{i,\tau}}} |(z-\mu_i)^{q_i-r_{i,\tau}}  +\e\g_{i,\tau}|
\\
=& C\e^{\frac{r_{i,\tau}}{q_i-r_{i,\tau}}} \prod\limits_{s=1}^{q_i-r_{i,\tau}}|z-z_s|
\geqslant
C\tilde{b}
\e^{ \frac{q_i+1}{q_i-r_{i,\tau}}},
\end{aligned}
\end{equation}
where $C$ is some constant independent of $\e$ and $\tilde{b}$.
We observe that under our assumptions
\begin{equation*}
\frac{q_i+1}{q_i-r_{i,\tau}}\leqslant 2.
\end{equation*}
In view of the latter inequality, on the boundary of the ball $T_{\tilde{b}\e^{\frac{2}{q_i-r_{i,\tau}}}}(z_j(\e))$, the estimates
\begin{equation*}
|\e (z-\mu_i)^{r_{i,\tau}+1}Y_3(z) + \e^2 Y_4(z,\e)|\leqslant C(\e^{\frac{q_i+1}{q_i-r_{i,\tau}}}+\e^2)<C\e^{\frac{q_i+1}{q_i-r_{i,\tau}}}
\end{equation*}
hold true, where $C$ is some constant independent of $\e$ and $\tilde{b}$. This estimate and (\ref{5.32}) allow us to apply the Rouch\'e theorem and we conclude that the ball $T_{\tilde{b}\e^{\frac{2}{q_i-r_{i,\tau}}}}(z_j(\e))$ contains exactly one zero of equation (\ref{5.23}). This gives asymptotic expansion (\ref{2.19}).
The proof is complete.

\section{Emerging resonances and eigenvalues}

In this section we determine whether the poles of the operators $\cR_{\e,\tau}$ described in the previous section correspond to the eigenvalues or resonances of the operator $\Op_\e$. In this way, we shall prove Theorems~\ref{thEmBot},~\ref{thEmInt}.

We fix  $p\in\{1,\ldots, m\}$ and assume that $\L_p=\ldots=\L_{p+n-1}$, where $n\geqslant1$ is a multiplicity of the eigenvalue $\L_p$ of the operator $\Op'$. In what follows, we analyze the nature of poles $k_{ij}(\e)$ from Theorem~\ref{thEmer1} with asymptotic behavior (\ref{2.10}). In order to do this, we analyze the behavior of the associated nontrivial solutions to problem (\ref{2.5}) with $f=0$.

\subsection{Bottom of the spectrum} In this subsection we assume that $p=1$ and we prove Theorem~\ref{thEmBot}. In the considered case, the  {meromorphic} 
continuation of the resolvent given the operator $\cR_{\e,\tau}$ is independent of $\tau$ and the same is true for the corresponding poles $k_{ij}(\e)$ from Theorem~\ref{thEmer1} with asymptotic behavior (\ref{2.10}). The behavior of the associated nontrivial solutions to problem (\ref{2.5}) with $f=0$  is provided by representation (\ref{2.8}). In this representation, the first sum over $j=1,\ldots,p-1$   is obviously missing and the leading term at infinity is the sum over $j=p,\ldots,p+n-1$.

We fix one of the poles $k_{ij}(\e)$ with some $i\in\{1,\ldots,n\}$, $j=1,\ldots,q_i$. It follows from (\ref{4.12}), (\ref{4.9a}), (\ref{4.9b}), (\ref{3.7}), (\ref{3.8}) that the coefficients $c_{s,\e}^\pm$, $s=1,\ldots,n$, in representation (\ref{2.8}) coincide with $l_s^\e$ determined by (\ref{4.12}). Since we deal with a non-trivial solution of system (\ref{4.12}), this implies that at least one of the coefficients $c_{s,\e}^\pm$ is non-zero. Then representation (\ref{2.8}) for a non-trivial solution to problem (\ref{2.5}) defined in (\ref{4.9a}) and formulae (\ref{4.11}) imply that if
\begin{equation}\label{6.1}
\RE k_{ij}(\e)>0,
\end{equation}
then a  non-trivial solution $\psi_\e$ decays exponentially at infinity and thus, is an eigenfunction, while the opposite inequality
\begin{equation}\label{6.2}
\RE k_{ij}(\e)\leqslant 0
\end{equation}
ensures that $\psi_\e$ is not in $W_2^2(\Om)$.

Asymptotic expansions for $k_{ij}$ established in Theorems~\ref{thEmer1},~\ref{thEmer2} allow us to check effectively the above conditions. Namely, if $\RE\mu_i>0$, this ensures inequality (\ref{6.1}) for sufficiently small $\e$.
Then the poles $k_{ij}$, $j=1,\ldots,q_i$ correspond to the eigenvalues $\l_{ij}(\e)=\L_p-k_{ij}^2(\e)$ with
asymptotic expansions (\ref{6.4}), (\ref{6.5}), (\ref{6.6}) as described in the formulation of Theorem~\ref{thEmBot}.
If $\RE\mu_i<0$, this guarantees inequality (\ref{6.2}) for sufficiently small $\e$ and the poles $k_{ij}$ correspond to the resonances $\l_{ij}(\e)=\L_p-k_{ij}^2(\e)$ with same asymptotic expansions (\ref{6.4}), (\ref{6.5}), (\ref{6.6}).

If $\RE\mu_i=0$, in order to understand which of   conditions (\ref{6.1}), (\ref{6.2}) is realized, we need to check the next term in the asymptotic expansion for $k_{ij}$. According Theorem~\ref{thEmer2}, this can be done under an additional assumption $2r_{i,\tau}\leqslant q_i-1$ and only for poles $k_{ij,\tau}$ with $j=r_i+1,\ldots,q_i$. Then asymptotic expansion (\ref{2.19}) imply that under condition (\ref{6.7}), the pole $k_{ij}$, $j=r_{i,\tau}+1,\ldots,q_i$ corresponds to an eigenvalue, while under condition (\ref{6.8})
the pole $k_{ij}$ corresponds to a resonance. The asymptotic behavior for this eigenvalue/resonance is given by (\ref{6.6}) if $\mu_i\ne0$ and it is given by (\ref{6.9}) if $\mu_i=0$.

\subsection{Internal thresholds in the spectrum} In this subsection we study the nature of the poles emerging from internal points $\L_p$ in the essential spectrum and we prove Theorem~\ref{thEmInt}. As in the previous subsection, here we again analyze the behavior at infinity of non-trivial solutions $\psi_\e$ associated with poles $k_{ij}$ and this analysis will be based on representation (\ref{2.8}). However, there are important differences. The matter is that now the representation involves also a sum over $j=1,\ldots,p-1$. Its terms  can decay or grow at infinity.
As it has been mentioned in the proof of Theorem~\ref{thAnCo}, according identity (\ref{4.11}), the functions $e^{-\tau K_j(k)|x_d|}$, $j=1,\ldots,p-1$, decay exponentially at infinity if $\tau\IM k^2<0$ and grow exponentially or vary periodically  as $\tau\IM k^2\geqslant 0$. This means that
apart of the sign of $\RE k_{ij}$, we should also control the sign of $\tau\IM k_{ij}^2$. One more point is that now we have two  {meromorphic} 
continuations of the resolvent, the operators $\cR_{\e,\tau}$, $\tau\in\{-1,+1\}$, respectively, two sets of their poles $k_{ij,\tau}(\e)$ converging to zero as $\e\to+0$.
We observe that according Theorems~\ref{thEmer1},~\ref{thEmer2}, the first terms in asymptotic expansions for poles $k_{ij,\tau}$ are the eigenvalues $\mu_i$ of the matrix $\rM_1$ and they are independent of $\tau$. In fact, we can see the influence of $\tau$ on the asymptotic behavior of $k_{ij}$ only in formulae (\ref{2.19}), for $j=r_{i,\tau}+1,\ldots,q_i$ under an additional assumption $2r_{i,\tau}\leqslant q_i-1$.

We choose $i\in\{1,\ldots,N\}$, $j\in\{1,\ldots,q_i\}$, $\tau\in\{-1,+1\}$ and consider the pole $k_{ij,\tau}(\e)$. As in the above proof of Theorem~\ref{thEmBot}, it is easy to see that at least one of the coefficients    $c_{s,\e}^\pm$, $s=1,\ldots,n$, in representation (\ref{2.8}) for the associated non-trivial functions $\psi_\e$ is non-zero. Then it follows from asymptotic expansions (\ref{2.13}), (\ref{2.17}), (\ref{2.18}), (\ref{2.19}) that conditions (\ref{2.26a}), (\ref{2.26b}), (\ref{2.27a}), (\ref{2.27b}) ensure  that $\RE k_{ij,\tau}(\e)>0$ and $\tau\IM k_{ij,\tau}(\e)<0$ and hence, the pole $k_{ij,\tau}(\e)$ corresponds to an eigenvalue $\l_{ij,\tau}(\e)=\L_p-k_{ij,\tau}^2(\e)$. The stated asymptotic behavior for this eigenvalue is implied by  (\ref{2.13}), (\ref{2.17}), (\ref{2.18}), (\ref{2.19}). In the same way, conditions (\ref{2.28a}), (\ref{2.28b}) ensure that $\RE k_{ij,\tau}(\e)<0$ and hence, the pole $k_{ij,\e}$ corresponds to a  resonance $\l_{ij,\tau}(\e)=\L_p-k_{ij,\tau}^2(\e)$ and it possesses the stated behavior.

If $\mu_i$ is a simple eigenvalue of the matrix $\rM_1$, then $q_i=1$. In this case, a non-trivial solution $l_\e$ to system (\ref{4.12}) associated with $k_{i1,\tau}(\e)$ converges to the vector $\ev_i$. Indeed, it follows from representation (\ref{5.20})  that the corresponding zero $z_{ij,\tau}(\e)$ of equation (\ref{5.6}) is an eigenvalue of the matrix $\rM_1-\e\tilde{\rM}_2(z_{ij,\tau}(\e),\e)$ and $l_\e$ is an associated eigenvector. And since   $z_{ij,\tau}(\e)$ converges to $\mu_i$ and both these eigenvalues are simple, the eigenvector  associated with $k_{i1,\tau}(\e)$ converges to $l_\e$ as well. Now by representation (\ref{4.9a}), (\ref{4.9b}), definition (\ref{3.9}) of the operator $\cG_{p,\tau}$ and formula (\ref{3.13}) we conclude that the coefficients $c_{s,\e}^\pm$ in representation (\ref{2.8}) satisfy the identities:
 \begin{equation*}
 c_{s,\e}^\pm=\e
\sum\limits_{t=1}^{n}\int\limits_{\Om} e^{\mp K_t(0)x_d} \overline{\psi_s(x')} \cL_1 \ev_{i,t} \psi_{t-p+1}\di x+o(\e).
 \end{equation*}
Hence, by condition (\ref{2.30}), at least one of the coefficients $c_{s,\e}^\pm$ is non-zero. As above, it is easy to see that one of  conditions (\ref{2.26a}) or (\ref{2.26b})  and one of conditions (\ref{2.29a}) or (\ref{2.29b}) ensures that the functions $e^{-K_s(k_{ij,\tau}(\e))|x_d|}$ grows exponentially at infinity and hence, the same is true for the function $\psi_\e$. Therefore, the pole $k_{ij,\tau}(\e)$ corresponds to a resonance. The asymptotic expansion for this resonance can be established as above. The proof of Theorem~\ref{thEmInt} is complete.

\section*{Acknowledgments}

{The authors thank the referees for useful remarks.}
The research by D.I.B. and D.A.Z. is supported by the Russian Science Foundation
(Grant No. 20-11-19995).

\section*{Data Availability Statement}

The data that support the findings of this study are available from the corresponding author upon reasonable request.

\end{document}